\documentclass[11pt,letterpaper]{article}
\usepackage[margin=1in]{geometry}
\usepackage[T1]{fontenc}
\usepackage{lmodern}
\usepackage{microtype}

\usepackage{hyperref}
\usepackage{amsmath,amsthm,amsfonts}
\usepackage{cleveref}
\usepackage{xspace}
\usepackage{url}
\usepackage{algorithm}
\usepackage{algpseudocode}
\usepackage{appendix}
\usepackage{mathtools}
\usepackage{bm}
\usepackage{paralist}
\usepackage{multirow}
\usepackage{booktabs}
\usepackage{graphicx}
\usepackage{subcaption}
\usepackage{thm-restate}
\usepackage[normalem]{ulem} 
\usepackage{xcolor} 
\usepackage[obeyFinal,textsize=footnotesize]{todonotes}

\hypersetup{colorlinks=true,allcolors=blue}

\theoremstyle{plain}
\newtheorem{theorem}{Theorem}[section]
\newtheorem{lemma}[theorem]{Lemma}
\newtheorem{claim}[theorem]{Claim}
\newtheorem*{claim*}{Claim}

\newtheorem{fact}[theorem]{Fact}

\theoremstyle{definition}
\newtheorem{definition}[theorem]{Definition}

\newtheorem{property}[theorem]{Property}

\theoremstyle{remark}

\newcommand{\ProblemName}[1]{\textsc{#1}}
\newcommand{\kCenter}{$k$-\ProblemName{Center}\xspace}
\newcommand{\kMedian}{$k$-\ProblemName{Median}\xspace}
\newcommand{\kMeans}{$k$-\ProblemName{Means}\xspace}

\def\RR{{\mathbb{R}}}

\def\ZZ{{\mathbb{Z}}}

\DeclareMathOperator{\poly}{poly}
\DeclareMathOperator{\cost}{cost}
\DeclareMathOperator{\OPT}{OPT}
\DeclareMathOperator{\APX}{APX}

\DeclareMathOperator{\dist}{dist}
\DeclareMathOperator{\diam}{diam}
\DeclareMathOperator{\MIS}{MIS}

\DeclareMathOperator{\MDS}{MDS}

\DeclareMathOperator*{\hash}{\varphi}
\DeclareMathOperator*{\rep}{rep}
\DeclareMathOperator*{\argmax}{argmax}

\newcommand{\IS}{\ensuremath{I}\xspace}

\newcommand{\Econd}{\ensuremath{\mathcal{E}_{\mathrm{cond}}}\xspace}
\newcommand{\Eextend}{\ensuremath{\mathcal{E}_{\mathrm{ext}}}\xspace}
\newcommand{\Acond}{\ensuremath{A_{\mathrm{cond}}}\xspace}
\newcommand{\Anew}{\ensuremath{A_{\mathrm{new}}}\xspace}
\newcommand{\conf}{\ensuremath{\mathrm{conf}}\xspace}

\newcommand{\Spi}[1]{\ensuremath{\mathcal{S}_{\pi_{#1}}}\xspace}

\newcommand{\Iseq}{\ensuremath{I^{\mathrm{seq}}}\xspace}

\newcommand{\Iloc}{\ensuremath{I^{\mathrm{loc}}}\xspace}

\newcommand{\proj}{\ensuremath{\mathrm{proj}}\xspace}

\renewcommand{\epsilon}{\ensuremath{\varepsilon}}

\let\epsilon\varepsilon

\title{Fully Scalable MPC Algorithms for Euclidean $k$-Center}

\author{
  Artur Czumaj\\
  University of Warwick\\
  \texttt{A.Czumaj@warwick.ac.uk}
  \and 
  Guichen Gao\\
  Peking University\\
  \texttt{gc.gao@stu.pku.edu.cn}
  \and 
  Mohsen Ghaffari\\
  MIT\\
  \texttt{ghaffari@mit.edu}
  \and 
  Shaofeng H.-C. Jiang\\
  Peking University\\
  \texttt{shaofeng.jiang@pku.edu.cn}
 }

\begin{document}

\phantomsection\addcontentsline{toc}{section}{Abstract}

\maketitle
\begin{abstract}
The \emph{$k$-center} problem is a fundamental optimization problem with numerous applications in machine learning, data analysis, data mining, and communication networks. The $k$-center problem has been extensively studied in the classical sequential setting for several decades, and more recently there have been some efforts in understanding the problem in parallel computing, on the Massively Parallel Computation (MPC) model. For now, we have a good understanding of $k$-center in the case where each local MPC machine has sufficient local memory to store some representatives from each cluster, that is, when one has $\Omega(k)$
local memory per machine. While this setting covers the case of small values of $k$, for a large number of clusters these algorithms require undesirably large local memory, making them poorly scalable. The case of large $k$ has been considered only recently for the \emph{fully scalable} low-local-memory MPC model for the Euclidean instances of the $k$-center problem. However, the earlier works have been considering only the constant dimensional Euclidean space, required a super-constant number of rounds, and produced only $k(1+o(1))$ centers whose cost is a super-constant approximation of $k$-center.

In this work, we significantly improve upon the earlier results for the $k$-center problem for the fully scalable low-local-memory MPC model. In the low dimensional Euclidean case in $\mathbb{R}^d$, we present the first constant-round fully scalable MPC algorithm for $(2+\varepsilon)$-approximation.
We push the ratio further to $(1 + \epsilon)$-approximation albeit using slightly more $(1 + \varepsilon)k$ centers.
All these results naturally extends to slightly super-constant values of $d$.
In the high-dimensional regime, we provide the first fully scalable MPC algorithm that in a constant number of rounds achieves an $O(\log n/ \log \log n)$-approximation for $k$-center.

\end{abstract}
 
\section{Introduction}
\label{sec:intro}

\noindent\textbf{Clustering and the \kCenter problem.} Clustering is a fundamental task in data analysis and machine learning.
We consider a well-known clustering problem, called \kCenter, in Euclidean spaces.
In this problem, given an integer parameter $k \geq 1$ and a dataset $P \subset \mathbb{R}^d$,
the goal is to find a \emph{center set} $C \subset \mathbb{R}^d$ of $k$ points,
such that the following clustering objective is minimized
\begin{equation}
    \label{eqn:kcenter_def}
    \cost(P, C) := \max_{p \in P} \dist(p, C).
\end{equation}
Here, $\dist(x, y) := \|x - y\|_2$ for any two points $x, y \in \mathbb{R}^d$, and $\dist(p, C) := \min_{c \in C} \dist(p, c)$ for any point $p\in \mathbb{R}^d$ and any set of points $C \subset \mathbb{R}^d$.

Solving \kCenter on massive data sets introduces outstanding scalability issues.
To meet this scalability challenge, practical approaches usually use several interconnected computers to solve the problem, i.e., they resort to distributed computing. Accordingly, there has been significant recent interest in scalable algorithms with provable guarantees for \kCenter~\cite{EneIM11,ImM15,MalkomesKCWM15,CeccarelloPP19,BEFM21,CCM23,HZ23,AG23,BBM23,BP24,LFWXW24}, primarily in the Massively Parallel Computing (MPC) model, which has nowadays become the de-facto standard theoretical model for such large-scale distributed computation settings (see, e.g., \cite{GSZ11,BKS17,IKLMV23}).

\paragraph{Massively Parallel Computation (MPC) model.}
The MPC model, introduced in~\cite{KarloffSV10}, provides a theoretical abstraction for widely-used practical frameworks such as MapReduce~\cite{DBLP:journals/cacm/DeanG08},
Hadoop~\cite{white2012hadoop}, Spark~\cite{DBLP:conf/hotcloud/ZahariaCFSS10}, and Dryad~\cite{DBLP:conf/eurosys/IsardBYBF07}.
In this model, we are given a set of machines, each with some given memory of size $s$ (also known as \emph{local memory}). At the beginning of computation, the input (which in our case is a set of $n$ data points from $\mathbb{R}^d$) is arbitrarily distributed among these machines, with the constraint that it must fit within each machine's local memory. (Hence we will require that the number of machines is $\Omega(n/s)$, for otherwise the input would not fit the system.) The MPC computation proceeds in \emph{synchronous rounds}.
In each round, first, each machine processes its local data and performs an arbitrary computation on its data without communicating with other machines. Then, at the end of each round, machines can communicate by exchanging messages, subject to the constraint that for every machine, the total size of the messages it sends or receives is $O(s)$.
When the algorithm terminates, MPC machines collectively output the solution.
The goal is to finish the computational task using as small as possible number of rounds.

\paragraph{Local memory regimes and full scalability.}
The central parameter determining the computational model is the size of the local memory $s$. Unlike the input size $n$, local memory is defined by the hardware provided and as such, one would like the relation between $s$ and $n$ to be as flexible as possible. Therefore an ideal MPC algorithm should be \emph{fully scalable}, meaning that it should work with $s = n^{\sigma}$ for any constant $\sigma \in (0,1)$. The importance of designing fully scalable MPC algorithms has been recently observed (cf. \cite{BhaskaraW18}) for clustering problems like \kCenter (and also $k$-means and $k$-median), where the prior research (see below) demonstrates that the problem's difficulty changes radically depending on whether $s = \Omega(k)$ or not, i.e., whether one machine can hold the entire set of proposed centers or not.
It is furthermore desirable for the algorithm to use a near-linear total memory.

\paragraph{Prior work with high local memory requirements.}
In the non-fully scalable regime, when $s = \Omega(k)$, a classical technique of \textit{coresets} \cite{Har-Peled04,Har-PeledM04} can be applied
to reduce the $n$ input points into a $(1 + \epsilon)$-approximate proxy with only $O(k)$ points (ignoring $f(d) \cdot \poly(\log n)$ factors).
For \kCenter, provided that $s = \Omega(k n^\gamma)$ for some constant $\gamma \in (0, 1)$ \cite{BBM23},
this small proxy can be computed in $O(1)$ MPC rounds
and moved to a single machine. The clustering problem (in fact, its approximate version because of the approximation caused by the use of coresets) can then be solved locally without further communication in a single MPC round.
However, the coreset approach is not applicable when $s = o(k)$, because coresets suffer a trivial size lower bound of $\Omega(k)$ making the approach sketched above unsuitable.
Hence, new techniques must be developed for \emph{fully scalable algorithms when local memory is sublinear in $k$}.

Several other works for \kCenter, although not using a coreset directly, also follow a similar paradigm of finding a sketch of $\poly(k)$ points in each machine~\cite{EneIM11,MalkomesKCWM15,HZ23,AG23}, and therefore they still require $s = \Omega(k)$.
We remark that these results work under general metrics with distance oracle,
which is a setting very different from our Euclidean setting.
This fundamental difference translates to different flavor of studies: in general metrics not much more can be done than using the triangle inequality, whereas in $\mathbb{R}^d$ we need to investigate what Euclidean properties are useful for fully scalable algorithms.

\paragraph{Toward fully scalable MPC algorithms.}
To combat these technical challenges, several recent works have designed fully scalable MPC algorithms for \kCenter and for related clustering problems, including \kMedian and \kMeans~\cite{BhaskaraW18,BEFM21,Cohen-AddadLNSS21,Cohen-AddadMZ22,CCM23,CGJKV24}. These MPC algorithms are fully scalable in the sense that they can work with $s = n^{\sigma}$ for any constant $\sigma\in (0,1)$. Indeed, they usually work with any $s \geq f(d) \poly\log n$ regardless of $k$ (albeit many of these results output more than $k$ centers, only achieving a bi-criteria approximation).
Therefore, in particular, despite the inspiring recent progress, fully scalable MPC algorithms for \kCenter are poorly understood.
The state-of-the-art algorithm (for geometric \kCenter and only for $d=O(1)$) is by Coy, Czumaj, and Mishra~\cite{CCM23}:
it achieves a \emph{super-constant} approximation ratio of $O(\log^* n)$ (improving over an $O(\log\log\log n)$-rounds bound from an earlier work of Batteni et al.\ \cite{BEFM21}),
using $k + o(k)$ centers; this violates the constraint of using at most $k$ centers and works in a \emph{super-constant} $O(\log\log n)$ number of rounds.
These bounds, which are stated for the standard regime of $s = n^\sigma$ with a constant $\sigma\in (0,1)$, show a drastic gap to the above-mentioned bounds achieved in the fully scalable $s = \Omega(k)$ regime.

\paragraph{Challenges of high dimensionality.}
Apart from the fully-scalability,
the high dimensionality of Euclidean spaces is another challenge in designing MPC algorithms.
Indeed, there has been emerging works that address high dimensionality in MPC
\cite{BhaskaraW18,Cohen-AddadLNSS21,Cohen-AddadMZ22,EpastoMMZ22,CGJKV24,JayaramMNZ24,AzarmehrBJLMZ25}.
Unfortunately, all previous fully-scalable MPC algorithms for \kCenter only work for low-dimensional regime since it requires $\exp(d)$ dependence in $d$ in local space,
and fully-scalable MPC algorithms suitable for high dimension, especially those with $\poly(d)$ dependence, constitutes an open area of research.
Overall, there has been a large gap in fully-scalable algorithms for \kCenter
under both low- and high-dimensional regime.

\subsection{Our Results}
\label{sec:result}
We give new fully scalable MPC algorithms for Euclidean \kCenter
and we systematically address both the low-dimensional and high-dimensional regimes.
Our results in low dimension significantly improve previous results simultaneously in various aspects (i.e., approximation ratio, round complexity, etc.).
We also obtain the first results for high dimension, and our approximation ratio bypasses several natural barriers.

\paragraph{Low-dimensional regime.}
In low dimension, we provide the first fully scalable MPC algorithm that achieves a constant approximation to \kCenter, running in a constant number of rounds (\Cref{thm:low_dim_2approx}).
This result significantly improves the previous fully scalable algorithms for \kCenter~\cite{BEFM21,CCM23} in several major aspects: by achieving constant round complexity, better approximation factor, and true approximation (without bi-criteria considerations that allow slightly more than $k$ centers).

\begin{restatable}{theorem}{thmmain}
    \label{thm:low_dim_2approx}
There exists an MPC algorithm that given $\epsilon\in (0,1)$, $k \ge 1$, and a dataset $P \subset \mathbb{R}^d$ of $n$ points  distributed across MPC machines with local memory $s \geq (\Omega(d \epsilon^{-1}))^{\Omega(d)} \poly\log n$,
with probability at least $1-1/n$ computes a $(2 + \varepsilon)$-approximate solution to \kCenter,
using $O(\log_s n)$ rounds and total memory $O(n \cdot \poly\log n \cdot (O(d \epsilon^{-1}))^{O(d)})$.
\end{restatable}

Furthermore, we can improve the $(2+\epsilon)$ approximation bound if we allow bi-criteria approximations: we design a $(1 + \epsilon)$-approximation \kCenter algorithm that uses $(1 + \epsilon)k$ centers (\Cref{thm:low_dim_bicrit}).
This bi-criteria approximation is almost optimal, and is the first of its kind in the literature.

\begin{theorem}
    \label{thm:low_dim_bicrit}
There exists an MPC algorithm that given $\epsilon\in (0,1)$, $k \ge 1$, and a dataset $P \subset \mathbb{R}^d$ of $n$ points  distributed across MPC machines with local memory $s \geq (\Omega(d \epsilon^{-1}))^{\Omega(d)} \poly\log n$,
with probability at least $1-1/n$ computes
a $(1 + \epsilon, 1 + \epsilon)$-approximate solution\footnote{An \emph{$(\alpha, \beta)$-approximate} solution to \kCenter (also called a \emph{bi-criteria} solution) is a center set $C \subset \mathbb{R}^d$ that has at most $\beta k$ centers and has cost at most $\alpha$ times the optimal cost of using at most $k$ centers.}
to \kCenter, using $O(\log_s n)$ rounds and total memory $O(n \cdot \poly\log n \cdot (O(d \epsilon^{-1}))^{O(d)})$.
\end{theorem}

Observe that for the memory regime of $s = n^{\sigma}$ with a constant $\sigma\in(0,1)$,
both \Cref{thm:low_dim_2approx,thm:low_dim_bicrit} run in a constant number of rounds. Moreover, $\Theta(\log_s n)$ is the complexity of far more rudimentary tasks, e.g., outputting the summation of $n$ numbers (see, e.g., \cite{RVW18}).
The dependence on $d$ in the memory bounds is $2^{\Theta(d \log d)}$.
Thus, for any constant $\delta \in (0,1)$, if $d = o(\log n/\log\log n)$, then the algorithms work in a constant number of rounds with local memory $s \ge n^{\delta}$ and total memory $n^{1+o(1)}$,
which is the regime studied in the previous works on fully scalable algorithms for \kCenter~\cite{BEFM21,CCM23}. Furthermore, if $d = o(\log\log n/\log\log\log n)$, then the total memory is in the desirable regime of $O(n \poly\log(n))$.

\paragraph{High-dimensional regime.}
Next, we go beyond the low-dimensional regime and explore the high dimension case where $d$ can be as large as $O(\log n)$.\footnote{As also observed in e.g.,~\cite{Cohen-AddadLNSS21,CGJKV24},
one can assume $d = O(\log n)$ without loss of generality by a Johnson-Lindenstrauss transform.
}
For this regime, we provide the first fully scalable MPC algorithm for \kCenter, and it achieves an $O(\log n/ \log \log n)$-approximation, running in a constant number of rounds (\Cref{thm:kcenter_high_dim}).

\begin{theorem}
\label{thm:kcenter_high_dim}
There exists an MPC algorithm that given $\varepsilon\in (0, 1)$,  $k\geq 1$, and a dataset $P\subset \mathbb{R}^{d}$ of $n$ points distributed across MPC machines with local memory $s\geq \poly(d\log n)$, with probability at least $1-1/n$ computes an $O(\epsilon^{-1}\log n/ \log \log n)$-approximate solution to \kCenter, using $O(\log_{s} n)$ rounds and total memory $O(n^{1 + \epsilon}\poly(d\log n))$.
\end{theorem}

We are not aware of any earlier fully scalable MPC algorithms for \kCenter in high dimension that we could compare with.
In fact, fully scalable MPC algorithms for clustering problems in high dimension are generally less understood,
and the only known result is an $O(\log^{2} n)$-approximation for \kMedian~\cite{Cohen-AddadLNSS21} (and in fact this ratio may be improved to $O(\log^{1,.5} n)$ using the techniques from~\cite{AhanchiAHKZ23} although it is not explicitly mentioned),
and as far as we know, nothing is known for \kCenter and \kMeans.
These existing results for \kMedian rely on the tree embedding technique,
and currently only an $O(\log^{1.5} n)$-distortion is known~\cite{AhanchiAHKZ23} (which translates to the ratio).
As a result, even if these techniques could be adapted for \kCenter,
it would only provide an $O(\log^{1.5} n)$-approximation, which falls short of the $O(\log n/ \log \log n)$-approximation achieved by our method;
in fact, our bound is even better than the fundamental lower bound of $\Omega(\log n)$-approximation of tree embedding.
This is not to mention the added technical difficulty of using the tree embedding: its expected distance distortion guarantee is too weak to be useful for the ``max'' aggregation of distances in \kCenter.

\subsection{Technical Overview}
\label{sec:tech_overview}

Our algorithms for \kCenter rely on variants of the classic reductions to geometric versions of the \emph{ruling set} (RS) and \emph{minimum dominating set} (MDS) problems,
which are fundamental problems in distributed computing.
We briefly describe the reductions in \Cref{subsec:intro_reductions}, particularly to mention our exact setup of RS and MDS, and state the results we obtain for each.
Then in \Cref{sec:intro_mis_and_mds} and \Cref{sec:intro_rs}, we provide a technical overview of our proposed MPC algorithms for these results, focusing on the key challenges and core techniques. 
While the relation between \kCenter and RS/MDS is well-known and has also been 
used in MPC algorithms for \kCenter in general metrics~\cite{HZ23,AG23}, it has not been studied for Euclidean \kCenter in the fully scalable setting.
This is a key technical difference to previous fully scalable algorithms for Euclidean \kCenter~\cite{BEFM21,CCM23} which employ successive uniform sampling to find centers. 

Both our low dimension and high dimension results rely on geometric hashing techniques (in \Cref{sec:hash}),
through which we utilize the Euclidean structure.
For low dimension, our algorithms are based on natural parallel algorithms where similar variants were also considered in graph MPC algorithms,
and the geometric hashing is the key to achieve the new bounds. 
For high dimension, our algorithm is a variant of the one-round version of Luby's algorithm. 
It has been known that the one-round Luby's algorithm yields a $\Theta(\log n)$ bound for RS in general graphs (see e.g.~\cite[Exercise 1.12]{dist_graph_book}). 
However, our new variant crucially makes use of the Euclidean structure via geometric hashing,
and it breaks the mentioned $\Theta(\log n)$ bound
in general graphs, improving it by an $O(\log\log n)$ factor in the Euclidean setting;
See \Cref{sec:intro_rs} for a more formal discussion.

\subsubsection{Reductions and Results for Geometric RS and MDS}
\label{subsec:intro_reductions}
To establish \Cref{thm:low_dim_2approx,thm:low_dim_bicrit,thm:kcenter_high_dim},
we begin with introducing the definitions and reductions for RS and MDS.
Let $\tau, \alpha>0$ be parameters, and let $\OPT$ be the minimum cost of the solution for \kCenter.

\paragraph{Geometric RS and MDS.}
A subset $S\subseteq P$ is called a $\tau$-independent set ($\tau$-IS) for $P$, if for every $x\neq y\in S$, $\dist(x, y)>\tau$,
and we say $S \subseteq \mathbb{R}^d$ is a $\tau$-dominating set ($\tau$-DS) for $P$,
if for every $x \in P$, $\dist(x, S) \leq \tau$.
A subset $S \subseteq P$ is a $(\tau, \alpha)$-ruling set ($(\tau, \alpha)$-RS) for $P$ if $S$ is both a $\tau$-IS and $\alpha$-DS for $P$.
A $\tau$-MDS is a $\tau$-DS with the minimum size, denoted as $\MDS_\tau(P)$.
A related well known notion is maximal independent set (MIS), where a $\tau$-MIS is $(\tau, \tau)$-RS.

\paragraph{Reductions.}
It is known (see, e.g., \cite{DBLP:journals/jacm/HochbaumS86}, and also \Cref{fact:is_ub}) that any $\tau$-IS  of $P$ for $\tau \geq 2\OPT$ must have at most $k$ points.
Therefore, an MPC algorithm that computes a $(2\OPT, \alpha)$-RS of $P$ would immediately yield $\alpha$-approximation for \kCenter on $P$.
On the other hand, for MDS,
since the optimal solution for \kCenter itself is a candidate for a $\tau$-MDS
and has at most $k$ points for $\tau = \OPT$,
a $(1 + \epsilon)$-approximation to $\tau$-MDS would yield $(1 + \epsilon)k$ centers.
This relation helps to obtain the desired bi-criteria approximation.
Compared with the setting of RS which could only leads to some $O(1)$-approximation,
MDS operates on the entire $\mathbb{R}^d$.
This is necessary for the $(1 + \epsilon)$ ratio since centers need to be picked from $\mathbb{R}^d$ instead of only from $P$.

\paragraph{Results for RS and MDS.}
We obtain the following results for RS and MDS, in both low and high dimension.
Combining with the abovementioned reductions, these results readily imply \Cref{thm:low_dim_2approx,thm:low_dim_bicrit,thm:kcenter_high_dim}.
All results run in $O(\log_s n)$ rounds which is constant in the typical setup of $s = n^{\sigma}$ for constant $0 < \sigma < 1$.
In low dimension, we obtain $(\tau, (1 + \epsilon)\tau)$-RS and a $(1 + \epsilon)\tau$-DS whose size is at most $(1 + \epsilon)$ times the $\tau$-MDS (which in a sense is a ``bi-criteria'' approximation).
Both results use $(\epsilon^{-1}d)^{O(d)} \cdot \poly(\log(n))$ local space,
and $n \poly(d\log n) \cdot (\epsilon^{-1}d)^{O(d)}$ total space.
In high dimension, we obtain $(\tau, (\epsilon^{-1} \log n / \log\log n)\tau)$-RS,
using (ideal) $\poly(d\log n)$ local space and $n^{1 + \epsilon} \poly(d \log n)$ total space.
These results are summarized in \Cref{tab:rs_mds}.

\begin{table}[t]
    \begin{tabular}{llll}
        \toprule
        guarantee & local space & total space & reference \\
        \midrule
        $(\tau, (1 + \epsilon)\tau)$-RS & $(\epsilon^{-1}d)^{O(d)} \cdot \tilde O(1)$ & $\tilde{O}(n)  \cdot (\epsilon^{-1}d)^{O(d)}$ & \Cref{thm:MIS} \\
        $(1 + \epsilon)\tau$-DS of size $(1 + \epsilon) |\MDS_\tau(P)|$ & $(\epsilon^{-1}d)^{O(d)} \cdot \tilde O(1)$ &  $\tilde{O}(n) \cdot (\epsilon^{-1}d)^{O(d)}$ & \Cref{thm:mds} \\
        $(\tau, O\left(\epsilon^{-1} \frac{\log n}{\log \log n}\right)\tau)$-RS & $\tilde{O}(1)$  & $\tilde{O}(n^{1 + \epsilon})$  & \Cref{thm:ruling_set} \\
        \bottomrule
    \end{tabular}
    \caption{RS and MDS results, where $\tilde{O}$ hides $\poly(d\log n)$ factor, all run in $O(\log_s n)$ rounds.}
    \label{tab:rs_mds}
\end{table}

We remark that
MIS, RS, and MDS are fundamental yet notoriously challenging problems in MPC. Existing studies on these problems are mostly under (general) graphs or general metric spaces, and they achieve worse bounds than ours, e.g., they need to use a super-constant number of rounds~\cite{Onak18,GhaffariU19,GhaffariGJ20,GLMPSSSUV23,GP24,JKPS25}, and/or are not fully scalable~\cite{CKPU23,HZ23,GP24}.
However, our results seem to suggest that these problems in Euclidean spaces behave very differently than in graphs/general metrics.
On the one hand, we obtain fully-scalable algorithms in both low and high dimension,
but on the other hand, our algorithms are only ``approximations'' to MIS and MDS; for instance, in low dimension, both our RS and MDS results have $(1 + \epsilon)$ factor off in the dominating parameter to MIS and MDS.
For RS/MIS and MDS without violating dominating parameter,
we are only aware of a line of research in distributed computing for growth-bounded graphs, see Schneider and Wattenhofer \cite{SW10}, which indirectly lead to $O(\log^*n)$-rounds fully scalable algorithms for MIS/MDS in $\mathbb{R}^d$ for constant $d$.
It is still open to design fully scalable algorithms for MIS, even in 2D, in \emph{constant} number of rounds.
In fact, this problem is already challenging on a 2D input set with diameter $O(\tau)$.
Nevertheless, our RS and MDS results suffice for approximations  for \kCenter.

\subsubsection{RS and MDS in Low Dimension}
\label{sec:intro_mis_and_mds}
The RS and MDS in low dimension starts with a rounding to a $\epsilon \tau / \sqrt d$-grid.
Specifically, we move each data point to the nearest $\epsilon \tau / \sqrt{d}$-grid point, whose coordinates are multiples of $\epsilon \tau / \sqrt{d}$, denoted as $P'$.\footnote{In our proof we actually need to use a slightly different rounding to make sure the image also belongs to $P$.}
Then we show that any $(\tau, \alpha \tau)$-RS to $P'$ yields a $(\tau, (1 + \epsilon)\alpha \tau)$-RS to the original dataset $P$,
and similarly, any $\tau$-DS to $P'$ yields a $(1 + \epsilon)\tau$-DS to $P$.
Hence, we can assume without loss of generality that $P$ is a subset of the said grid,
and find RS and MDS on it.
This rounding is useful,
since it ensures that in any subset of $\mathbb{R}^d$ of diameter $\gamma \cdot \tau$ ($\gamma \geq 1$),
the number of the grid points is at most $(O(d \gamma / \epsilon))^d$.
Hence, as long as $s \geq \Omega(d / \epsilon)^{d}$, we can afford to bring all grid points in a small area to a single machine and solve RS/MDS on it locally.

\paragraph{An overview for the proof of RS.}
In fact, on the rounded dataset $P$, we find a $(\tau, \tau)$-RS which is as well $\tau$-MIS on $P$.
A standard way of finding MIS in a graph is a greedy algorithm: start with the entire vertex set,
and if the current vertex set is nonempty, then find any vertex $x$, add it to the output (MIS) set,
remove all vertices that are adjacent to $x$ from the current vertex set, and repeat.
In the geometric setting, we can use an improved version where in each iteration
we identify a large number of vertices that are independent, instead of only one.
Specifically,
we partition the entire $\mathbb{R}^d$ into some $T$ \textit{groups} $\mathcal{W}_{1}, \ldots, \mathcal{W}_{T}$,
such that each group consists of \textit{regions} that are $\tau$ apart from each other.
Furthermore, each region has a bounded diameter $\alpha \tau$ for some $\alpha \geq 1$. \footnote{The reader may notice the reminiscence with the widely-used notion of network decomposition in graphs~\cite{awerbuch1989network,rozhovn2020polylogarithmic}. In that notion, the node set $V$ of the graph is partitioned into $T=O(\log n)$ groups $V_1$, \dots, $V_T$, such that in the subgraph induced by each $V_i$, each connected component has diameter at most $\alpha=O(\log n).$}
For $d = 2$, there is a simple way to partition with $T = O(1)$ and $\alpha = O(1)$,
as illustrated in \Cref{fig:decomposition}.
For general $d$, we give a partition that achieves $T = O(d)$ and $\alpha = O(d^{1.5})$. See \Cref{lemma:hash} for the precise statement.

\begin{figure}[ht]
    \centering
    \includegraphics[width=0.3\textwidth]{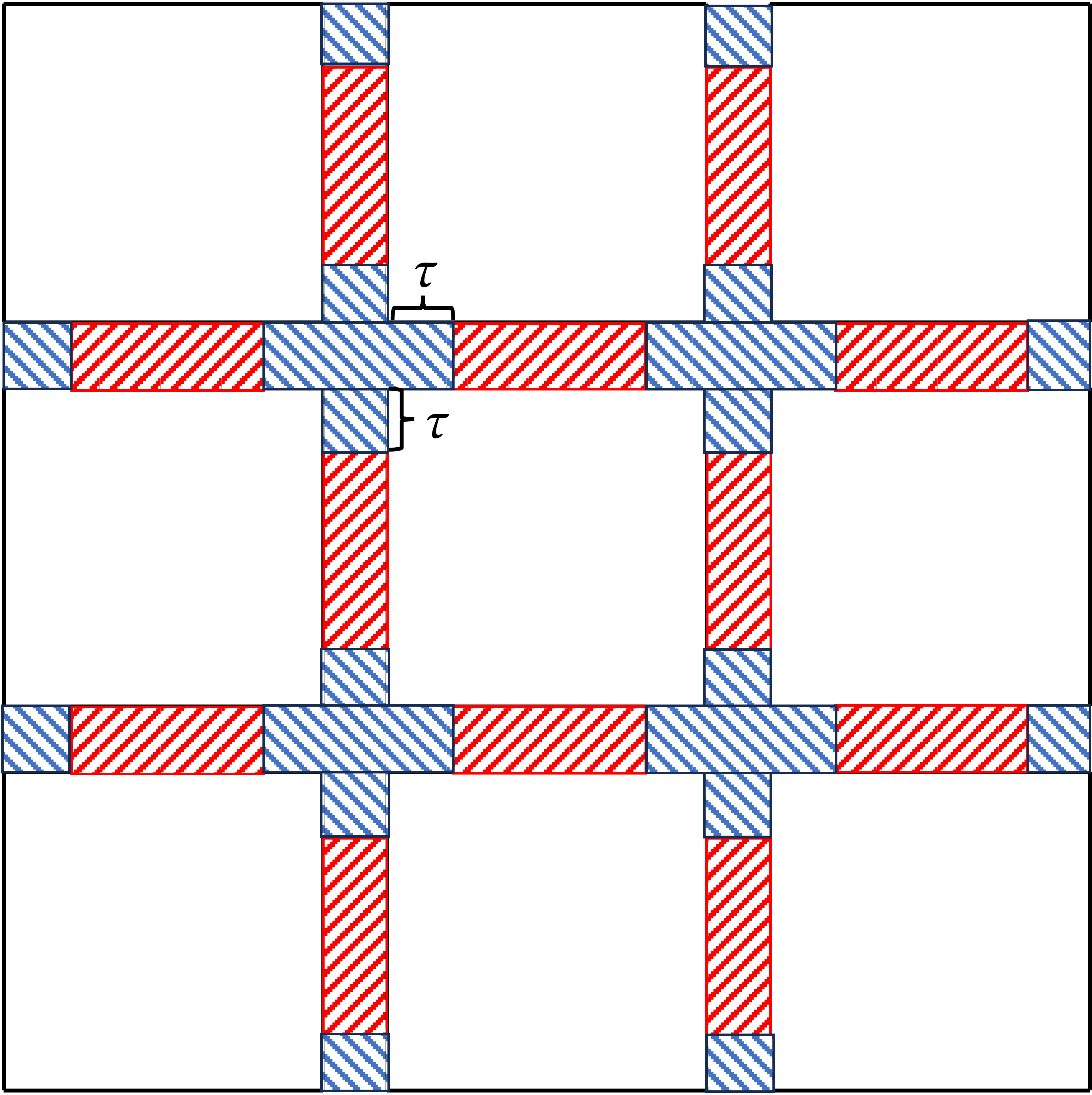}
    \caption{A space partition in 2D with $T = 3$ and $\alpha = 5$.
    The first group is squares with side-length $5\tau/\sqrt{2}$ (the blank space), the second is the red-shaded rectangles, and the third is the cross-like structures in blue shades.}
    \label{fig:decomposition}
\end{figure}

A key property of this decomposition is that the MIS computation in each region can be done independently within each group, as regions are $\tau$-separated.
This yields an $O(T\log_{s}n)$-round MPC algorithm: iterating over the $T$ groups $\mathcal{W}_{1}, \ldots, \mathcal{W}_{T}$, computing an MIS in each region in parallel, and removing points within $\tau$ from selected MIS points before proceeding to the next group.
Since each region in any group is of diameter $\alpha\tau$ which is at most $d^{1.5}\tau$,
there are at most $O(\varepsilon^{-1}d)^{O(d)}$ data points in each region, using the assumption of the rounded instance.
Hence, as long as $s \geq \Omega(\varepsilon^{-1}d)^{\Omega(d)}$,
the data points in each region can be stored in a single machine. Each such iteration can be implemented in $O(\log_s n)$ MPC rounds, and thus the overall scheme which has $T$ iterations takes $O(T\log_s n)$ MPC rounds.

To reduce this to $O(\log_{s} n)$ rounds, we exploit the locality of the above procedure.
Instead of iterating sequentially over groups, each region $R\in \mathcal{W}_{i}$ for $1\leq i\leq T$  directly determines its final subset $R'$ by identifying the influence from earlier groups $\bigcup_{j<i}\mathcal{W}_{j}$ that are within a bounded range $O(i\tau + (i-1)\alpha\tau)\leq \poly(d)\cdot \tau$.
Since an MIS selection only affects a $\tau$-radius per iteration, and each region has a diameter at most $\alpha\tau$, the cumulative affected area over $i$ iterations remains bounded. Leveraging the rounding property again, the number of relevant regions remains at most $O(\varepsilon^{-1}d)^{O(d)}$, allowing all necessary regions to be replicated locally.
This ensures that each region can compute its MIS in parallel in $O(\log_{s}n)$ rounds, provided $s\geq \Omega(\varepsilon^{-1}d)^{\Omega(d)}\poly(\log n)$.

\paragraph{An overview for the proof of MDS.}
We start with a weaker local space bound that requires a $2^{\Omega(d^2)}$ dependence of $d$ in $s$,
instead of the claimed $2^{\Omega(d \log d)}$ bound.
Our approach starts with a simple algorithm:
partition $\mathbb{R}^d$ into hypercubes of side-length $\alpha \tau$, where $\alpha \geq 1$ is a parameter to be determined. Then send the data points in each hypercube to a single machine, and solve locally the exact MDS of each hypercube.
Taking the union of these MDSs forms the final dominating set. Clearly, this returns a dominating set for the entire dataset,
but it may not be of small size.
Intuitively, the major gap in this algorithm is when the optimal solution uses points located at the boundary of the hypercubes to dominate the adjacent hypercubes,
while the algorithm described here only uses a point to dominate points in a single hypercube.

To address this issue, we observe that bridging the gap between the algorithm's solution and the optimal solution requires only bounding the number of points in the outer ``$\tau$-extended region'' of each hypercube $R$, denoted as $U_{\tau}^{\infty}(R)$, that intersect with the optimal solution.
Specifically, it suffices to show that this number is at most an $\varepsilon$-fraction of the optimal solution size.
To establish this bound, we apply an averaging argument:
If we shift the entire hypercube partitioning by some multiple of $O(\tau)$, there must exist a shift that satisfies our requirement, provided that the hypercube side-length $\alpha\tau$ is sufficiently large.
This may be visualized more easily in 1D:
each $U_{\tau}^{\infty}(R)$ is simply two intervals of length $\tau$ to the left and right of $R$ (which itself is also an interval, whose length is $\alpha \tau$).
Then by shifting a multiple of $O(\tau)$, the two intervals of $U_\tau^{\infty}$,
 after these shifts,
form a partition of the entire $\mathbb{R}$.
Unfortunately, this simple shifting of hypercubes only leads to $\alpha = 2^{O(d)}$,
which translates to a $2^{O(d^2)}$ dependence of $d$ in the local space $s$.
The main reason for this $\alpha = 2^{O(d)}$ bound is that the same point in the optimal solution may belong to up to $2^{O(d)}$ sets $U_\tau^{\infty}(R)$ (for some hypercube $R$).

To further reduce $\alpha = \poly(d)$, which leads to the $d^{O(d)}$ local space bound,
we need to employ a more sophisticated geometric hashing. We state this in \Cref{lemma:hash} and \Cref{fact:consistent_hash}.
This hash maps each point in $\mathbb{R}^d$ into some bucket,
and we would replace the hypercubes in the abovementioned algorithm with such buckets.
An important property of this bucketing is that any point in $\mathbb{R}^d$ (hence any point in the optimal solution)
can intersect at most $\poly(d)$ number of sets $U_\tau^{\infty}(R)$ over all buckets $R$, instead of $2^{O(d)}$ as in the simple hypercube partition.
However, the use of this new bucketing also introduces additional issues.
Specifically, since the buckets are of complicated structure,
it is difficult to analyze the even more complex set $U^\infty_\tau(R)$ for the averaging argument.
To this end, we manage to show (in \Cref{lemma:hash}, third property) that the union of $\bigcup_R U^\infty_\tau(R)$
is contained in the complement of a Cartesian power of (1D) intervals (e.g., $([1, 2] \cup [3, 4])^d$).
This structure of Cartesian power is similar enough to the $U_{\tau}^{\infty}$ annulus of hypercubes (which may be handled by projecting to each dimension),
albeit taking the complement.
This eventually enables us to use a modified averaging argument to finish the proof.

\subsubsection{RS in High Dimension}
\label{sec:intro_rs}
Our $(\tau, O(\epsilon^{-1}\log n / \log\log n)\tau)$-RS in high dimension is a modification of the well-known Luby's algorithm.
In this discussion, we assume $\epsilon = \Theta(1)$ and ignore this parameter.
The first modification is that,
unlike the standard Luby's algorithm which runs for $O(\log n)$ iterations~\cite{Luby85},
our algorithm only runs Luby's for one iteration:
for every data point $x \in P$, generate a uniform random value $h(x) \in [0, 1]$,
and then for each $x\in P$, include $x$ in the RS if $x$ has the smallest $h$ value in $B_P(x, \tau)$ (the ball centered at $x$ with radius $\tau$, intersecting points in $P$). 
This one-round Luby's algorithm achieves $(\tau, O(\log n)\tau)$-RS (with high probability), and we also show that this is tight in general graphs; see \Cref{sec:one_round_luby}.\footnote{Similar bounds were also mentioned without proof in the literature, see e.g.~\cite[Exercise 1.12]{dist_graph_book}.
We give a proof (sketch) for this tight bound for completeness.}
However, this is worse than the $O(\log n / \log \log n)$ factor that we can achieve.

\paragraph{A new preprocessing step based on geometric hashing.}
Hence, we need to introduce the second important modification to this one-round Luby's algorithm in order to bypass the $O(\log n)$ factor.
Specifically, before running the one-round Luby's algorithm,
we run a preprocessing step to map the data points to the buckets of a geometric hashing.
The geometric hashing that we use is the \emph{consistent hashing}~\cite{CJKVY22} (see \Cref{lemma:consist_hash}),
and the concrete guarantee in our context is that, each hash bucket has diameter $\ell := O(\log n / \log \log n)\tau$,
and for any subset $S \subseteq \mathbb{R}^d$ with diameter at most $O(\tau)$,
the number of buckets that $S$ intersects is at most $\Lambda := \poly(\log n)$.
We pick an arbitrary point in $P$ from each (non-empty) bucket, denoting the resultant set as $P'$,
and we run the one-round Luby's on $P'$.
Clearly, this hashing step only additively increases the dominating parameter by $\ell = O(\log n / \log\log n)\tau$ which we can afford.
At a high level, the use of this rounding is to limit the size of the $O(\tau)$-neighborhood for every point,
which is analogue to the degree of a graph.
In a sense, what we prove is that one-round Luby's on graphs with degree bound $\poly(\log n)$ yields an $O(\log n / \log\log n)$-ruling set.

Next, we explain in more detail why this hashing step helps to obtain $O(\log n / \log\log n)\tau$ dominating parameter.
This requires us to do a formal analysis to one-round Luby's algorithm (which did not seem to appear in the literature),
and utilize the property of hashing that for every $x\in P'$, $|B_{P'}(x, \tau)| \leq \Lambda = \poly(\log n)$.

\paragraph{A re-assignment argument.}
Let $R$ be the resultant set found by running one-round Luby's on $P'$.
Fix a point $p \in P'$, and we need to upper bound $\dist(x, R)$.
To this end, we interpret the algorithm as defining an auxiliary sequence $S = (x_0 := p, x_1, \ldots, x_T)$ (where $T$ is also random).
Specifically, $S$ is formed in the following process: we start with $i = 0$,
whenever $x_{i}$ is not picked into $R$ we define $x_{i + 1}$ as the point with the smallest $h$ value in $B_{P'}(x, \tau)$,
and we terminate the procedure otherwise.
Clearly, $\dist(x, R) \leq T \tau$, and it suffices to upper bound $T$ (which is a random stopping time).
Indeed, a similar plan of re-assignment argument has been employed in~\cite{CGJKV24},
but the definition and analysis of $S$ is quite different since they focus on facility location.

Now, a crucial step is to bound the probability of the event that, given the algorithm picks a prefix $(x_0, \ldots, x_i)$ of sequence $S$,
the algorithm picks some new $x_{i + 1} \in P'$ instead of terminating.
We call this \emph{extension} probability.
Ideally, if we can give this extension probability a universal upper bound $\gamma < 1$,
then for $t \geq 1$,
the probability that $T > t$ is at most $\gamma^t$, which decreases exponentially with respect to $t$.
However, this seemingly good bound may not immediately be sufficient,
since one still needs to take a union bound over sequences of length $t$, because the sequence $S$ is random.
A na\"ive bound is $n^t$ since each $x_i$ may be picked from any point in $P'$, and this is positively large (i.e., $\gamma^t$ cannot ``cancel it out''.).
Moreover, an added difficulty is that the ``ideal'' case of having universal upper bound $\gamma$ for the extension probability may not be possible in the first place.

Our analysis resolves both difficulties.
We show in \Cref{lemma:ext_cond} that the extension probability is roughly upper bounded by $|B_{P'}(x_i, \tau)| / |\bigcup_{j = 1}^i B_{P'}(x_j, \tau)|$ (recalling that $(x_0, \ldots, x_i)$ is given).
For the union bound,
since the hashing guarantees that $|B_{P'}(x, \tau)| \leq \Lambda = \poly(\log n)$ for any $x \in P'$,
we can approximate the size $|B_{P'}(x_i, \tau)|$ by rounding to the next power of $2$, denoted as $\lceil |B_{P'}(x_i, \tau)| \rceil_2$,
and this yields only $O(\log\log n)$ possibilities after rounding.
For a (fixed) sequence $S' := (x_0, \ldots, x_m)$,
we define its \emph{configuration} as the rounded value of $(\lceil |B_{P'}(x_1, \tau)| \rceil_2, \ldots, \lceil |B_{P'}(x_m, \tau)| \rceil_2)$.
We can do a union bound with respect to the configuration of the sequences,
and there are at most $O(\log \log n)^t$ number of configurations for length-$t$ sequences.

Finally, given a sequence $S' = (x_0, \ldots, x_t)$ with a given configuration (for some $t$),
to upper bound the extension probability,
we need to analyze $\prod_{i} \frac{|B_{P'}(x_i, \tau)|}{|\sum_{j = 1}^i B_{P'}(x_j, \tau)|}$,
and we show in a key \Cref{lemma:key} that this is upper bounded by $\exp(-\Omega(t) \log(t / \log \Lambda))$.
Therefore, we can pick $t = \log n / \log \log n$,
apply the union bound and conclude that $\Pr[T > t] \leq (\log \log n)^t \cdot \exp(-\Omega(t) \log(t / \log \Lambda)) \leq 1 / \poly(n)$.

We also provide a (sketch) of how to slightly modify our analysis to show the one-round Luby's algorithm yields $O(\tau, O(\log n)\tau)$-RS, in \Cref{sec:one_round_luby}. As mentioned, this did not seem to appear in the literature.
We also provide a tight instance for this one-round Luby's algorithm, in \Cref{sec:luby_graph_lb}.

     \subsection{Related Work}

The \kCenter problem, as one of the fundamental clustering problems \cite{Gonzalez85,HS85,DBLP:journals/jacm/HochbaumS86}, has been studied extensively in sequential setting, and more recently, also in parallel setting. For MPC algorithms for \kCenter, most of prior works have focused on non-fully scalable algorithms that have a dependence of $\Omega(k)$ in the local memory $s$ and can generally achieve $O(1)$ rounds. 
In general metric space,~\cite{EneIM11} obtains a large-constant approximation,
using $O(1 / \sigma)$ rounds if the local memory is $\poly(k)n^\sigma$,
which offers a tradeoff between the number of rounds and the local memory.
More recent works aim to achieve a smaller constant ratio, down to factor 2, at the cost of having local memory size with a fixed polynomial dependence in $n$ and/or a polynomial dependence in the number of machines~\cite{MalkomesKCWM15,ImM15,HZ23,AG23}.
There has been also some \kCenter studies on doubling metrics instances (which is a generalization of $\mathbb{R}^d$)~\cite{CeccarelloPP19,BBM23}, where in the bi-criteria (or with outliers) setting, 
not only a constant but even a $(1+\epsilon)$ ratio can be achieved due to the low-dimensional structure~\cite{BBM23}.

Fully-scalable MPC algorithms have been also studied for related clustering problems of \kMedian and \kMeans in $\mathbb{R}^d$. \cite{Cohen-AddadLNSS21} describes a hierarchical clustering algorithm that in a constant number of rounds returns a $\poly\log (n)$-approximation for \kMedian using $s = \poly(d) \cdot n^{\sigma}$ local memory; we are not aware of any similar approximation for \kMeans (even when $d = O(1)$ and with $\poly\log n$ ratio). Furthermore, since the techniques used in \cite{Cohen-AddadLNSS21} rely critically on the approach of hierarchically well-separated trees, they are unlikely to lead to sub-polylogarithmic approximation ratio algorithms.
On the other hand, bi-criteria approximation are known for both \kMedian and \kMeans~\cite{BhaskaraW18,CGJKV24}, and in a constant number of rounds and with $s = \poly(d) \cdot n^{\sigma}$ local memory, one can approximate $k$-median and $k$-means with $O(1)$ ratio using $(1 + \epsilon)k$ centers~\cite{CGJKV24}. In the same setting, it is also possible to achieve a $(1 + \epsilon)$-approximation for special inputs~\cite{Cohen-AddadMZ22}.

     \section{Preliminaries}
\label{sec:prelim}

For integer $n \geq 1$, let $[n] := \{1, \ldots, n\}$. 
For a mapping $f : X \to Y$, denote $f^{-1}(y) := \{x \in X: f(x) = y\}$ as the pre-image of $f$. 
For some $d \geq 1$, a set $S \subseteq \mathbb{R}^d$ and $x \in \mathbb{R}^d$,
write $S + x$ to denote $\{ x + y : y \in S \}$.
For $t > 0$, let $\mathcal{G}_t \subseteq \mathbb{R}^d$ be the set of $t$-grid points, that is, points whose coordinates take values as integer multiples of $t$. 
For a subset $S \subset \mathbb{R}^d$, let $\diam(S) := \max_{x, y\in S} \dist(x, y)$ be its diameter.
Similarly define $\dist_\infty$ and $\diam_\infty$ as the $\ell_\infty$ version.

\paragraph{Neighborhoods.}
Let $B(x, r):=\{y\in \mathbb{R}^{d}: \dist(x,y)\leq r\}$ be a ball centered at $x\in \mathbb{R}^{d}$ with radius $r\geq 0$.
For a point set $X\subseteq\mathbb{R}^{d}$, let $B_{X}(x, r):= B(x, r)\cap X$ denote a ball inside $X$.
For a point set $S\subset \mathbb{R}^{d}$ and $\gamma>0$, let $N_{\gamma}^{\infty}(S)$ be the $\gamma$-neighborhood of $S$ under the $\ell_{\infty}$ distance, i.e., 
$$N_{\gamma}^{\infty}(S):= \{x\in\mathbb{R}^{d}: \exists s\in S, ||x-s||_{\infty}\leq \gamma\},$$ 
and let $U_{\gamma}^{\infty}(S):= N_{\gamma}^{\infty}(S)\setminus S$ be the $\gamma$-annulus of $S$ under $\ell_{\infty}$ distance.

\paragraph{Geometric independent set, ruling set and dominating set.} 
Let $\tau>0$ be some threshold, $\alpha>0$ be some parameter and $P\subset \mathbb{R}^{d}$ be a dataset. 
A subset $S\subseteq P$ is called a $\tau$-independent set ($\tau$-IS) for $P$, if for every $x\neq y\in S$, $\dist(x, y)>\tau$, 
and we say $S \subseteq \mathbb{R}^d$ is a $\tau$-dominating set ($\tau$-DS) for $P$,
if for every $x \in P$, $\dist(x, S) \leq \tau$. 
A subset $S \subseteq P$ is a $(\tau, \alpha)$-ruling set ($(\tau, \alpha)$-RS) for $P$ if $S$ is both a $\tau$-IS and $\alpha$-DS for $P$. 
A $\tau$-MDS is a $\tau$-DS with the minimum size, denoted as $\MDS_\tau(P)$. 
A related well known notion is maximal independent set (MIS), where a $\tau$-MIS for $P$, denoted as $\MIS_{\tau}(P)$, is a $(\tau, \tau)$-RS for $P$.

\paragraph{\kCenter.}
Recall that the objective of \kCenter is defined in \Cref{sec:intro} as $\cost(P, C)$ for a dataset $P$ and center set $C$.
Let $\OPT(P)$ be the minimum value of the solution for \kCenter, i.e., $\OPT(P):= \min_{C\subset \mathbb{R}^{d}} \cost(P, C)$. When the context is clear, we simply write $\OPT$ for $\OPT(P)$.

\paragraph{Standard MPC primitives.}

In our algorithms we frequently use several basic primitives on MPC with local memory $s \geq \poly\log(N)$ using total memory $O(N \poly\log N)$ and number of rounds $O(\log_s N)$, where $N$ is the size of a generic input.
This includes standard procedures of broadcast and converge-cast (of a message that is of size $\leq \sqrt s$), see e.g.~\cite{mpc_book}.
Goodrich et al.~\cite{GSZ11} show that the task of sorting $N$ numbers can also be performed deterministically in the above setting. 

\begin{lemma}[Packing property, cf. {\cite[Lemma 4.1]{pollard1990empirical}}]
    \label{lemma:packing}
    For a point set $S\subset \mathbb{R}^{d}$ such that $\forall x \neq y \in S, dist(x, y) \geq \rho$, we have that $|S|\leq (\frac{3\diam(S)}{\rho})^{d}$.  
\end{lemma}

     \section{Geometric Hashing}
\label{sec:hash}

We present our new geometric hashing in \Cref{lemma:hash}.
This hashing is crucially used in our low dimension results.
The construction of this hashing is the same as~\cite[Theorem 5.3]{CJKVY22},
and the first two properties have also been established in~\cite{CJKVY22}.
However, the third property is new, and is based on a careful analysis that utilizes the structure of this specific construction.
\begin{restatable}{lemma}{lemmahash}
    \label{lemma:hash}
    For every $\beta>0$, $\ell \geq \Theta(d^{1.5} \beta)$, there is a hash function $f : \mathbb{R}^{d} \to \mathbb{R}^{d}$ such that the following holds.
    \begin{enumerate}
        \item  Each bucket has diameter at most $\ell$, namely, for every image $u \in f(\mathbb{R}^{d})$, $\diam(f^{-1}(u)) \leq \ell$.
        \item  The bucket set $\{ f^{-1}(u) : u \in f(\mathbb{R}^d) \}$ can be partitioned into $d + 1$ groups $\{ \mathcal{W}_i \}_{i=0}^d$,
        such that every two buckets $S \neq S'$ in the same group $\mathcal{W}_i$ $(0 \leq i \leq d)$ has $\dist(S, S') \geq \dist_\infty(S, S') > \beta$.
\item  For every $0<\tau\leq \beta$, $\left(\bigcup_{u \in f(\mathbb{R}^{d})} U^\infty_\tau(f^{-1}(u)) \right) \cap L(z, 2b)^d = \emptyset$,\footnote{Recall that the notation $L(\cdot, \cdot)^d$ denotes the $d$-th Cartesian power of $L(\cdot, \cdot)$.}
    where $z:=\ell/ \sqrt{d}$, $b := d\beta +\tau$ and $L(p, q):= \bigcup_{a\in \ZZ} [ap + q, (a+1)p - q)$ for $p > 2q$.
    \end{enumerate} 
    Furthermore, it takes $\poly(d)$ space to store $f$ and to evaluate $f(x)$ for every $x \in \mathbb{R}^d$.
\end{restatable}
\begin{proof}
    The proof can be found in \Cref{sec:proof_hashing}.
\end{proof}

\Cref{lemma:hash} readily implies the following property.
Roughly speaking, it ensures that for any point set $S$ with small enough $\ell_\infty$ diameter, the number of intersected buckets in the hash is bounded.
\begin{fact}
    \label{fact:consistent_hash}
    The second property of \Cref{lemma:hash}
    implies that for every $S \subset \mathbb{R}^d$ such that $\diam_\infty(S) \leq \beta$,
    it holds that $|f(S)| \leq d + 1$.
\end{fact}

Geometric hash functions that has similar guarantee as in \Cref{fact:consistent_hash} has been studied under the notion of consistent hashing~\cite{CJKVY22}, or sparse partitions which interpret the hashing as a space partition~\cite{JLNRS05,Filtser24}.
Specifically, the definition of consistent hashing is stated as follows.
The consistently guarantee of \Cref{fact:consistent_hash} is slightly stronger in the sense that the diameter bound of $S$ is in $\ell_\infty$ instead of $\ell_2$.
\begin{definition}[{\cite[Definition 1.6]{CJKVY22}}]
    \label{def:consistent_hash}
    A mapping $\varphi: \mathbb{R}^{d} \to \mathbb{R}^{d}$ is called a $\Gamma$-gap $\Lambda$-consistent hash with diameter bound $\ell>0$, or simply $(\Gamma, \Lambda)$-hash , if it satisfies: 
    \begin{itemize}
        \item Diameter: for every image $z\in \varphi(\mathbb{R}^{d})$, we have $\diam(\varphi^{-1}(z))\leq \ell$; and 
        \item Consistency: for every $S\subset \mathbb{R}^{d}$ with $\diam(S)\leq \ell/ \Gamma$, we have $|\varphi(S)|\leq \Lambda$. 
    \end{itemize}
\end{definition}
\Cref{lemma:consist_hash} gives a space-efficient consistent hashing with near-optimal parameter tradeoffs.
We rely on this parameter tradeoff in a preprocessing step (\Cref{alg:rs_hash}) of our high dimension ruling set result \Cref{thm:ruling_set}.
\begin{lemma}[{\cite[Theorem 5.1]{CJKVY22}}]
    \label{lemma:consist_hash}
    For every $\Gamma\in [8, 2d]$, there exists a (deterministic) $(\Gamma, \Lambda)$-hash $\varphi: \mathbb{R}^{d} \to \mathbb{R}^{d}$
    where $\Lambda = \exp(8d/\Gamma) \cdot O(d\log d)$. 
    Furthermore, $\varphi$ can be described using $O(d^{2}\log^{2}d)$ bits and one can evaluate $\varphi(x)$ for any point $x\in \mathbb{R}^{d}$ in space $O(d^{2}\log^{2}d)$. 
\end{lemma}
     \section{MPC Algorithms for RS in Low Dimension}
\label{sec:MIS}
\begin{lemma}
    \label{thm:MIS}
    There is a deterministic MPC algorithm that given threshold $\tau>0$, constant $\varepsilon\in (0, 1)$ and dataset $P\subseteq \mathbb{R}^d$ of $n$ points distributed across MPC machines with local memory $s\geq \Omega(\varepsilon^{-1}d)^{\Omega(d)}\cdot \poly(\log n)$, 
computes a $(\tau, (1 + \epsilon)\tau)$-RS for $P$ in $O(\log_{s}n)$ rounds, using total memory $O(n\poly(d\log n))\cdot O(\varepsilon^{-1}d)^{O(d)}$. 
\end{lemma}

We start with an efficient MPC reduction, stated in \Cref{lemma:rounding}, that turns the input $P$ to a point set that satisfies the following property.
\begin{property}
    \label{prop:packing}
    A point set $P$ satisfies the property, if $\forall x \in \mathbb{R}^d$, $\rho \geq 1$, $B_P(x, \rho\tau) \leq (\epsilon^{-1}\rho d)^{O(d)}$.
\end{property}
\begin{restatable}{lemma}{lemmarounding}
    \label{lemma:rounding}
    There exists an MPC algorithm that computes a point set $P' \subseteq P$ such that $P'$ satisfies \Cref{prop:packing},
    and that any $\tau$-DS for $P'$ is a $(1 + O(\epsilon))\tau$-DS for $P$,
    within round and space bound as in \Cref{thm:MIS}.
\end{restatable}
    \begin{proof}
        To define $P'$, let $\phi : P \to \mathcal{G}_{\epsilon \tau / \sqrt{d}}$,
        such that for every $x \in P$, $\phi(x)$ maps to the nearest point in $\mathcal{G}_{\epsilon \tau / \sqrt{d}}$ of $x$.
        For every $z \in \phi(P)$, let $\rep(z)$ be an arbitrary fixed point in $\phi^{-1}(z) \cap P$.
        Then, define $P' := \{ \rep(z) : z \in \phi(P) \} \subseteq P$.
        Clearly, this whole process can be done within the claimed round and space as in \Cref{thm:MIS}.
        
        Now fix a $\tau$-DS $S$ for $P'$, and we verify that $S$ is $(1 + \epsilon)\tau$-DS for $P$.
        Observe that for $z \in \phi(P)$, we have $\dist(z, \rep(z)) \leq \epsilon \tau$.
        Then consider $x \in P$, we have
        \begin{align*}
            \dist(x, S)
            &\leq \dist(x, P') + \max_{y \in P'} \dist(y, S) \\
            &= \min_{z \in \phi(P)} \dist(x, \rep(z)) + \tau \\
            &\leq \min_{z \in \phi(P)} \{ \dist(x, z) + \dist(z, \rep(z)) \} + \tau \\
            &\leq (1 + 2\epsilon) \tau.
        \end{align*}
        Finally, to bound $|B_{P'}(x, \rho\tau)|$, observe that
        \begin{align*}
            |B_{P'} (x, \rho \tau)|
            &= |B(x, \rho \tau) \cap \{ \rep(z) : z \in \phi(P) \}| = | \{ z \in \phi(P) : \dist(x, \rep(z)) \leq \rho\tau \} |,
        \end{align*}
        where the second equality follows from the fact that $z \mapsto \rep(z)$ (for $z \in \phi(P)$) is one-to-one correspondence.
        Let $Q :=  \{ z \in \phi(P) : \dist(x, \rep(z)) \leq \rho\tau \} $, and it suffices to bound $|Q|$.
        To this end, we wish to apply \Cref{lemma:packing}.
        Observe that $Q \subseteq \mathcal{G}_{\epsilon \tau / \sqrt d}$.
        To bound the diameter of $Q$, consider $u \neq v \in Q \subseteq \phi(P)$, then by triangle inequality
        \begin{align*}
            \dist(u, v)
            &\leq \dist(u, \rep(u)) + \dist(\rep(u), x) + \dist(x, \rep(v)) + \dist(\rep(v), v) \\
            &\leq 2\epsilon \tau +  2\rho \tau\\
            &= O(\rho)\tau.
        \end{align*}
        Hence, we conclude that $|B_{P'}(x, O(\tau))| = |Q| \leq (\epsilon^{-1} \rho d)^{O(d)}$ by \Cref{lemma:packing}.
        This finishes the proof.
    \end{proof}

Consider a $P'$ as in \Cref{lemma:rounding},
since a $\tau$-MIS on $P'$ is also a $\tau$-DS for $P'$,
then it implies that a $\tau$-MIS on $P'$ is a $(\tau, (1 + O(\epsilon))\tau)$-RS for $P$.
Therefore,  in the remainder of the proof,
it suffices to design an MPC algorithm that finds a $\tau$-MIS on a point set $P$ with the assumption that $P$ satisfies \Cref{prop:packing} (and we omit the notation $P'$).

To proceed, our MIS algorithm relies on the following notion of geometric decomposition.

\begin{definition}
    \label{def:decomposition}
    A collection $\mathcal{W}:=\left\{S_{1}, S_{2}, \ldots\right\}$  of (disjoint) point sets of $\mathbb{R}^{d}$ is an $\alpha$-bounded $\beta$-separated decomposition, or simply $(\alpha, \beta)$-decomposition, if it satisfies the following.
    \begin{enumerate}
        \item For every $S\in \mathcal{W}$, $\diam(S)\leq \alpha$.
        \item For every distinct $S_{i}, S_{j}\in \mathcal{W}$, $\dist(S_{i}, S_{j})>\beta$.
    \end{enumerate}
    In addition, we say that a set of points $X\subseteq \mathbb{R}^{d}$ admits an $(\alpha, \beta)$-decomposition, if $X$ can be \emph{partitioned} into some $(\alpha, \beta)$-decomposition $\mathcal{W}$.
\end{definition}

\begin{fact}
    \label{fact:partition}
    For threshold $\tau > 0$, let $f$ be the hash guaranteed by \Cref{lemma:hash} with parameter $\ell = O(d^{1.5}\tau)$.
    Then by the second property of \Cref{lemma:hash},
    the buckets of the hash can be partitioned into $d + 1$ groups $\{\mathcal{W}_i\}_{i=0}^d$, 
    such that the collection of the buckets $\mathcal{W}_i$ is an $(O(d^{1.5})\tau, \tau)$-decomposition.
\end{fact}

\paragraph{Proof overview.} 
Before we present the final algorithm, we start with a more sequential version which is listed in \Cref{alg:seq_MIS}, whose correctness is easily analyzed in \Cref{lemma:seq_MIS_correctness}. 
This algorithm relies on the hash guaranteed by \Cref{fact:partition} whose buckets can be partitioned into $d + 1$ groups.
It processes these $d + 1$ groups $\mathcal{W}_0, \ldots, \mathcal{W}_d$ one by one, where in a group $\mathcal{W}_i$ the MIS of all buckets $S \in \mathcal{W}_i$ are computed.
Our final $O(1)$-round algorithm is a parallel implementation of it, where we utilize the locality of buckets among $\mathcal{W}_i$'s.

\begin{algorithm}
\caption{Basic sequential $\tau$-MIS algorithm for the set $P$ satisfying \Cref{prop:packing}}
\label{alg:seq_MIS}
\begin{algorithmic}[1]

    \State let $f$ be the hash and let $\mathcal{W}_0, \ldots, \mathcal{W}_d$ be the groups of buckets from \Cref{lemma:hash} and \Cref{fact:partition}
    
    \Comment{so each $\mathcal{W}_i$ is an $(O(d^{1.5})\tau, \tau)$-decomposition}
      
    \For{$i \gets 0, \ldots, d$}
    \For{$S \in \mathcal{W}_i$}
    \State let $\Iseq_{i, S} \gets$ $\MIS_\tau(S \cap P \setminus B(\Iseq_{i - 1}, \tau))$
    \label{line:Iseq_MIS_def}
    \Comment{define $\Iseq_{-1} \gets \emptyset$}
    \EndFor
    \State let $\Iseq_i \gets \Iseq_{i - 1} \cup \bigcup_{S \in \mathcal{W}_i}\Iseq_{i, S}$
    \label{line:Iseq_i_def}
    \EndFor
    \State return $\Iseq_d$
    \label{line:Iseq_d_def}
\end{algorithmic}
\end{algorithm}

\begin{lemma}
\label{lemma:seq_MIS_correctness}
    \Cref{alg:seq_MIS} returns a $\tau$-MIS for $P$. 
\end{lemma}
\begin{proof}
    The algorithm processes each $\mathcal{W}_{i}$ sequentially, remove the points that are within distance $\tau$ from the current independent set $\Iseq_{i-1}$ from $\mathcal{W}_{i}$, and finds MIS on the remaining points of buckets in $\mathcal{W}_{i}$. 
    It is clear that the found sets $\Iseq_{i}$ for $i= 0, \ldots, d$ are all independent, and in particular the returned set $\Iseq_{d}$ is independent. 
    Furthermore, since all $\mathcal{W}_{i}$'s are processed and that they form a partition of $\mathbb{R}^{d}$, $\Iseq_{d}$ is maximal. 
    This finishes the proof. 
\end{proof}

Notice that every distinct buckets $S, S' \in \mathcal{W}_i$ are $\tau$ apart.
Hence, in \Cref{alg:seq_MIS}, the MIS $\Iseq_{i,S}$ may be computed in parallel for all $S \in \mathcal{W}_i$.
However, these MIS also needs to know $\Iseq_{i - 1}$, which depends on the buckets in $\mathcal{W}_0, \ldots, \mathcal{W}_{i - 1}$,
and this requires to (sequentially) examining $\mathcal{W}_0, \ldots, \mathcal{W}_{i - 1}$ before computing $\Iseq_{i ,S}$.

A natural idea for avoiding this sequential computing, is to prune the buckets in sets $\mathcal{W}_0, \ldots, \mathcal{W}_{i - 1}$, so that only those ``relevant'' to $S$ are kept. 
Specifically, we are to figure out subsets $\mathcal{W}_j' \subseteq \mathcal{W}_j$ (which depends on $i$ and $S$) to replace $\mathcal{W}_j$, such that it suffices to simulate the MIS on buckets in $\mathcal{W}_j'$ for computing $\Iseq_{i, S}$. 
We shall see that these subsets $\mathcal{W}_j'$ are not only of small size, i.e., $O(\varepsilon^{-1}d)^{O(d)}$, but also can be easily identified since they consist of buckets that are ``close enough'' to $S$.

\paragraph{Making use of locality of $\mathcal{W}_i$'s.}
We implement this plan and derive a local version of \Cref{alg:seq_MIS}, listed in \Cref{alg:loc_MIS}. 
This algorithm aims to simulate the MIS, $\Iseq_{i, S}$, for every $0 \leq i \leq d$ and $S \in \mathcal{W}_i$. 
The difference to \Cref{alg:seq_MIS} is that it only uses the restricted 
$\mathcal{W}_j' = \{ S' \in \mathcal{W}_j : \dist(S', S) \leq O(d^{2.5}) \cdot \tau \}$ for $j \leq i - 1$ in replacement of buckets in $\mathcal{W}_j$. 
Next, we show that this choice of  $\mathcal{W}_j' = \{ S' \in \mathcal{W}_j : \dist(S', S) \leq O(d^{2.5}) \cdot \tau \}$ indeed suffices for simulating \Cref{alg:seq_MIS} locally.

\begin{algorithm}[ht]
\caption{Localized $\tau$-MIS algorithm for the set $P$ satisfying \Cref{prop:packing}}
\label{alg:loc_MIS}
\begin{algorithmic}[1]

    \State let $f$ be the hash and let $\mathcal{W}_0, \ldots, \mathcal{W}_d$ be the groups of buckets from \Cref{lemma:hash} and \Cref{fact:partition}
    
    \Comment{so each $\mathcal{W}_i$ is an $(O(d^{1.5})\tau, \tau)$-decomposition}
    
    \For{$i \gets 0, \ldots, d$}
    \For{$S\in \mathcal{W}_i$}
    \State for $j \leq i - 1$, let $\mathcal{W}'_j \gets \{ S' \in  \mathcal{W}_j : \dist(S', S) \leq O(d^{2.5}) \cdot \tau \}$ be the relevant set of buckets from $\mathcal{W}_j$
    \label{line:Wjprime}
    
    \State let $\Iloc_0 \gets \emptyset$ 
    \For{$j \gets 0, \ldots, i - 1$}
    \State let $\Iloc_j \gets \Iloc_{j - 1} \cup \bigcup_{S' \in \mathcal{W}_j^{'}} \MIS_{\tau}(S'\cap P\setminus B(\Iloc_{j-1}, \tau))$ 
    \label{line:Iloc_i_def}
    \EndFor
    \State let $\Iloc_{i, S} \gets \MIS_\tau(S \cap P \setminus B(\Iloc_{i - 1}, \tau))$
    \label{line:Iloc_MIS_def}
    \EndFor
    \EndFor
    \State return $\bigcup_{0\leq i \leq d, S \in \mathcal{W}_i} \Iloc_{i, S}$
    \label{line:Iloc_d_def}
\end{algorithmic}
\end{algorithm}

\begin{lemma}
\label{lemma:couple-seq-loc}
    For every $0 \leq i \leq d$, $S \in \mathcal{W}_i$, we have $\Iseq_{i, S} = \Iloc_{i, S}$, where $\Iseq$ is as in \Cref{alg:seq_MIS} and $\Iloc$ is as in \Cref{alg:loc_MIS}.
\end{lemma}
\begin{proof}
Consider some $0 \leq i \leq d$ and $S \in \mathcal{W}_i$. 
By definition, $\Iseq_{i, S} = \MIS_\tau(S \cap P \setminus B(\Iseq_{i-1}, \tau)) = \MIS_\tau(S \cap P \setminus B( \bigcup_{j \leq i-1, S' \in \mathcal{W}_j} \Iseq_{j, S'}, \tau))$. 
By the first equality, we know that $\Iseq_{i, S} \subseteq S$, and by the second equality, we can see that $\Iseq_{i, S}$ only depends on $\Iseq_{j, S'}$'s with $ j \leq i-1, S' \in \mathcal{W}_j$ such that  $B(\Iseq_{j, S'}, \tau) \cap S \neq \emptyset$, which is equivalent to $\dist(\Iseq_{j, S'}, S) \leq \tau$. 

Now, apply this argument again on each abovementioned $\Iseq_{j, S'}$, it holds that such $\Iseq_{j, S'}$ only depends on those $\Iseq_{j', S''}$ with $j' \leq j - 1, S'' \in \mathcal{W}_{j'}$ such that $\dist(S'', S') \leq \tau$ which further implies $\dist(S'', S) \leq 2\tau + \diam(S')\leq 2\tau + O(d^{1.5})\tau$, where the last inequality holds by \Cref{fact:partition}. 
Apply this argument recursively, one can eventually conclude that $\Iseq_{i, S}$ only depends on $\Iseq_{j, S'}$'s such that $j \leq i - 1, S' \in \mathcal{W}_j$ and $\dist(S', S) \leq i \cdot \tau + (i-1)O(d^{1.5})\tau\leq O(d^{2.5})\tau$, since there is at most $d+1$ groups.

Therefore, to evaluate $\Iseq_{i, S}$, one only needs to know buckets $\mathcal{W}_j' = \{ S' \in \mathcal{W}_j : \dist(S', S) \leq O(d^{2.5}) \cdot \tau \}$, simulate \Cref{alg:seq_MIS} on $\mathcal{W}_j'$ for $j \leq i - 1$ to obtain some $I'$, remove from $S$ the points that are within $I'$, and find $\MIS_\tau(S \cap P \setminus B(I', \tau))$. 
Finally, to obtain $\Iseq_{i, S} = \Iloc_{i, S}$, it remains to ensure that the MIS on the same subset in both \Cref{alg:seq_MIS} and \Cref{alg:loc_MIS} are the same/consistent.
Luckily, this can be trivially achieved since one can use a deterministic greedy maximal independent set algorithm for the purpose. 
This finishes the proof of \Cref{lemma:couple-seq-loc}. 
\end{proof}

\begin{lemma}
\label{lemma:relevant_buckets}
    For each $0\leq j\leq d-1$, $\sum_{S'\in \mathcal{W}_{j}'} |S'\cap P|\leq O(\varepsilon^{-1}d)^{O(d)}$ where $\mathcal{W}'_{j}$ is as in \Cref{alg:loc_MIS} line~\ref{line:Wjprime}.  
\end{lemma}
\begin{proof}
Consider some $0\leq i\leq d$ and $S\in \mathcal{W}_{i}$. 
Fix a $j\leq i-1$. 
Recall that $\mathcal{W}'_{j} = \{S'\in \mathcal{W}_{j}: \dist(S', S)\leq O(d^{2.5})\tau \}$ as in line~\ref{line:Wjprime}. 
Since the buckets in the same $\mathcal{W}_{j'}$ for every $j'$ are at least $\tau$ apart, we have that $\sum_{S'\in \mathcal{W}_{j}'} |S'\cap P| = |\bigcup_{S'\in \mathcal{W}_{j}'} (S'\cap P)|$. 
Fix any two points $x\neq y$ that come from buckets $S_{1}, S_{2}\in \mathcal{W}_{j}'$ respectively. 
Recall that each bucket has a diameter $O(d^{1.5})\tau$. 
Then, by the triangle inequality and the definition of $\mathcal{W}_{j}'$, we have that 
$$\dist(x, y)\leq \diam(S_{1}) + \dist(S_{1}, S) + \diam(S) + \dist(S, S_{2}) + \diam(S_{2})\leq O(d^{2.5})\tau.$$
This means that the diameter of $\bigcup_{S'\in \mathcal{W}_{j}'} S'$ is at most $O(d^{2.5})\tau$. 
Then, by \Cref{prop:packing}, we have that $|\bigcup_{S'\in \mathcal{W}_{j}'} (S'\cap P)| = |\bigcup_{S'\in \mathcal{W}_{j}'} S' \cap P|\leq O(\varepsilon^{-1}d)^{O(d)}$.
This finishes the proof of \Cref{lemma:relevant_buckets}.  
\end{proof}

\paragraph{MPC implementation.}
    Our final MPC algorithm is an MPC implementation of \Cref{alg:loc_MIS}.
    We do the two for-loops, i.e., for $i$ and $S$, in parallel,
    and run the steps inside the for-loops, i.e., the steps to compute $\Iloc_{i, S}$, only on a single machine.
    This procedure yields the same output as \Cref{alg:loc_MIS}, whose correctness has been analyzed in the above.

\paragraph{Space and round analysis.}
    The local space usage is dominated by
    computing $\Iloc_{i, S}$ (for some $i$ and $S \in \mathcal{W}_i$) in a single machine,
    which is dominated by $O(d\cdot \sum_{S'\in \mathcal{W}_j'}|S'\cap P|)$, for $\mathcal{W}_j'$ in line~\ref{line:Wjprime}. 
    By \Cref{lemma:relevant_buckets} one concludes that for each $0\leq j\leq d-1$, $\sum_{S'\in \mathcal{W}_j'}|S'\cap P| \leq O(\varepsilon^{-1}d)^{O(d)}$. 
This concludes the local space requirement.
    This also implies the total space is bounded by 
    $O(n\poly(d\log n))\cdot O(\varepsilon^{-1}d)^{O(d)}$, since the computation of $\mathcal{W}_j'$ requires to copy data points. 
    For the round complexity,
    since we do parallel-for loops, the round complexity is dominated by the steps inside each loop.
    These steps can be done in $O(\log_s n)$ rounds,
    since $\mathcal{W}_j'$ can be evaluated locally using the (data-oblivious) hash $f$,
    and other steps mostly require standard procedure of fetching the points from a certain bucket to a machine.

    This finishes the proof of \Cref{thm:MIS}. 
    \qed

    \section{MPC Algorithms for Approximate $\MDS$ in Low Dimension}
\label{sec:mds}

\begin{restatable}{lemma}{thmmds}
    \label{thm:mds} 
    There is a deterministic MPC algorithm that given threshold $\tau > 0$,  parameter $\varepsilon\in(0,1)$ and dataset $P\subseteq \mathbb{R}^d$ of $n$ points distributed across MPC machines with local memory $s\geq \Omega(\varepsilon^{-1}d)^{\Omega(d)}\poly(\log n)$, 
computes a $(1 + \epsilon)\tau$-DS $S \subset \mathbb{R}^{d}$ for $P$ such that $|S| \leq (1 + \epsilon) |\MDS_{\tau}(P)|$ in $O(\log_{s}n)$ rounds, using total memory $O(n \poly(\log n))\cdot O(\varepsilon^{-1}d)^{O(d)}$. 
\end{restatable}
We make use of the same reduction as in \Cref{sec:MIS}, specifically \Cref{lemma:rounding},
so that we can assume without loss of generality that $P$ satisfies \Cref{prop:packing}.

\paragraph{Proof overview.}
We give an outline of our MPC algorithm in \Cref{alg:imp_approx_mds}. 
The high level idea is to use the hash $f$ as in \Cref{lemma:hash},
compute locally a $\MDS$ for each bucket, and then take the union of this as the approximate solution.
This simple plan may not work well when data points are located around the ``boundary'' of the buckets,
since the optimal solution may use only a few ``boundary'' points near the buckets to dominate the points inside the buckets.

To combat this issue, we plan to shift the buckets (implemented by hash $f_v$ in line~\ref{line:f_v}),
and use an averaging argument to show there exists some good shift $v$,
such that the boundary points contribute to the optimal solution very little.
This averaging argument requires to use both a larger bucket diameter
$\ell = O(d^{3.5}\epsilon^{-1} \beta)$ (as opposed to the default $\ell = O(d^{1.5} \beta)$ as in \Cref{lemma:hash}),
and a large enough support $\mathcal{V}^d$ of the shift.
Both $\ell$ and $|\mathcal{V}^d|$ affect the space usage, and it is crucial to give good upper bounds for them.
A key property is \Cref{fact:consistent_hash}, and this helps us to use only about $(d\epsilon^{-1})^{O(d)}$ number of shifts.

\begin{algorithm}[ht]
    \caption{Algorithm outline for $(1 + \epsilon)$-approximate $\MDS_{\tau}(P)$}
    \label{alg:imp_approx_mds}
    \begin{algorithmic}[1]
        \State let $\beta:= 2\tau$, $\ell:= O(d^{3.5}\varepsilon^{-1}\beta)$, $z:= \ell/\sqrt{d}$, $b:= d\beta+\tau$, $T:= z/4b$, and $\mathcal{V} := \{ 0, 4b, \ldots, 4b(T-1) \}$
        \label{line:parameters}
       
        \State let $f: \mathbb{R}^{d}\to \mathbb{R}^{d}$ be a hash function with parameter $\beta$  and $\ell$ as guaranteed by \Cref{lemma:hash} 

        \label{line:f}

        \For{$v \in \mathcal{V}^d$ in parallel}
        \label{line:for-loop}
        
        \State define $f_v : \mathbb{R}^d \to \mathbb{R}^d$ such that $f_{v}(x):= f(x + v)$ for $x \in \mathbb{R}^d$
        \label{line:f_v}
         
        \State for each bucket $u\in f_{v}(P)$, compute $\MDS_{\tau}(f_{v}^{-1}(u)\cap P)$
        \label{line:mds_bucket}
            
        \State let $D_{v}\gets \bigcup_{u\in f_{v}(P)} \MDS_{\tau}(f_{v}^{-1}(u)\cap P)$
        \label{line:mds_aggregate}

        \EndFor
        \State return $\widehat{M} := D_{v}$ such that $|D_{v}|$ is the minimum over all $v \in \mathcal{V}^d$
        \label{line:return}
           
    \end{algorithmic}
\end{algorithm}

\paragraph{MPC implementation.}
We discuss how to implement \Cref{alg:imp_approx_mds} in MPC.
By \Cref{lemma:hash}, the hash in line~\ref{line:f} is deterministic and only takes $\poly(d)$ space,
so it can be locally stored in every machine without communication.
Next, we examine the lines in the parallel for-loop.
Line~\ref{line:f_v} can be executed locally in each machine.
Line~\ref{line:mds_bucket} can be implemented by first sorting the points with respect to their hash value (i.e., bucket ID),
then solving the $\MDS$ for each subset of points with the same hash value.
Notice that since each bucket has diameter at most $\ell = O(d^{3.5}\epsilon^{-1} \beta) \leq O((d^{3.5}\epsilon^{-1})\tau)$,
it contains at most $O(\varepsilon^{-1}d)^{O(d)}$ points from $P$ by \Cref{prop:packing}, 
so all these points can fit in the same machine (due to the assumption on $s$).
Therefore, their $\MDS$ can be solved locally without further communication.
Line~\ref{line:mds_aggregate} requires removing duplicated elements in $\MDS_\tau(f_v^{-1}(u) \cap P)$ for $u \in f_v(P)$,
which can be done using sorting.
Finally, line~\ref{line:return} can be implemented using standard procedures, specifically counting and finding the minimum element.

\paragraph{Round and space complexity.}
By the above, the round complexity is $O(\log_s n)$ since every step can be implemented in $O(\log_s n)$ rounds and the for-loop is done in parallel.
The local space requirement is dominated by line~\ref{line:mds_bucket},
which is $O(\varepsilon^{-1}d)^{O(d)}\cdot \poly(\log n)$.
The total space is dominated by the parallel-for loop,
which has $|\mathcal{V}|^d \leq (d\epsilon^{-1})^{O(d)}$ invocations,
and this results in $O(n \poly(\log n)) \cdot O(\varepsilon^{-1}d)^{O(d)}$ total space.

\paragraph{Error analysis.}
Let $M^* := \MDS_\tau(P)$ be an optimal/minimum solution of $\tau$-dominating set for $P$.
Observe that the algorithm returns $\widehat{M}$,
such that its size $|\widehat{M}| = \min_{v \in \mathcal{V}^d} |D_v|$.
Hence, it suffices to show there exists $v^* \in \mathcal{V}^d$ such that $|D_{v^*}| \leq (1 + \epsilon) |M^*|$.
To do so, we start with analyzing $|D_v|$ for a generic $v$ as follows,
in order to obtain a more concrete condition of $v^*$.
\begin{align}
    |D_v|
    &= \left|\bigcup_{u\in f_{v}(P)} \MDS_{\tau}(f_{v}^{-1}(u)\cap P)\right| \nonumber \\
    &\leq \sum_{u \in f_v(P)} |\MDS_\tau(f_v^{-1}(u) \cap P)| \nonumber \\
    &\leq \sum_{u \in f_v(P)} |M^* \cap N_\tau^\infty(f_v^{-1}(u))| \nonumber \\
    &\leq \sum_{u \in f_v(P)} |M^* \cap f_v^{-1}(u)| + |M^* \cap U_\tau^\infty(f_v^{-1}(u))| \nonumber \\
    &= |M^*| + \sum_{u \in f_v(P)} |M^* \cap U_\tau^\infty(f_v^{-1}(u))|. \label{eqn:DvtoMstar}
\end{align}
We next further analyze $\sum_{u \in f_v(P)} |M^* \cap U_\tau^\infty(f_v^{-1}(u))|$,
and the plan is to use the properties of \Cref{lemma:hash}.
However, those properties are for $f$ instead of directly for $f_v$, and hence we need the following \Cref{fact:f_fv} that relates $f_v^{-1}$ and $f^{-1}$.
\begin{fact}
    \label{fact:f_fv}
    For every $v \in \mathbb{R}^d$, 
    the image set of $f_v$, i.e., $f_v(\mathbb{R}^d)$, equals that of $f$,
    and for every image $u \in f_v(\mathbb{R}^d)$, $f_v^{-1}(u) = f^{-1}(u) - v$.
\end{fact}
By \Cref{fact:f_fv},
\begin{equation}
    \label{eqn:remove_v}
    \sum_{u \in f_v(P)} |M^* \cap U_\tau^\infty(f_v^{-1}(u))|
    \leq \sum_{u \in f_v(\mathbb{R}^d)} |M^* \cap U_\tau^\infty(f_v^{-1}(u))|
    \leq \sum_{u \in f(\mathbb{R}^d)} |M^* \cap (U_\tau^\infty(f^{-1}(u)) - v)|.
\end{equation}
We next apply \Cref{fact:consistent_hash} to further upper bound \eqref{eqn:remove_v}, i.e.,
\begin{equation}
    \label{eqn:back_to_union}
    \sum_{u \in f(\mathbb{R}^d)} |M^* \cap (U_\tau^\infty(f^{-1}(u)) - v)|
    \leq 
    (d + 1) \cdot \left| \bigcup_{u \in f(\mathbb{R}^d)} M^* \cap (U_\tau^\infty(f^{-1}(u)) - v) \right|.
\end{equation}
    To see this, for $x \in \mathbb{R}^d$,
    if some $u$ satisfies $x \in U_\tau^\infty(f^{-1}(u))$,
    then there exists $y \in f^{-1}(u)$ such that $\|x - y\|_\infty \leq \tau$,
    which is equivalent to $N_\tau^\infty(x) \cap f^{-1}(u) \neq \emptyset$.
    Since $\diam_\infty(N_\tau^\infty(x)) \leq 2\tau = \beta$,
    by \Cref{fact:consistent_hash},
    $ |\{u : N_\tau^\infty(x) \cap f^{-1}(u) \neq \emptyset \}|
    = |f(N_\tau^\infty(x))| \leq d+ 1$.
    Therefore,
    every point $x \in \mathbb{R}^d$ can only participate in at most $ d + 1$ terms of
    $\sum_{u \in f(\mathbb{R}^d)} |M^* \cap (U_\tau^\infty(f^{-1}(u)) - v)|$,
    and this implies \eqref{eqn:back_to_union}.

Then, 
applying \Cref{fact:f_fv} and the third property of \Cref{lemma:hash} to the $U_\tau^\infty$'s of \eqref{eqn:back_to_union},
we have
\begin{align*}
    \bigcup_{u \in f(\mathbb{R}^d)} (U_\tau^\infty(f^{-1}(u)) - v)
    \subseteq \overline{L(z, 2b)^d} - v,
\end{align*}
where $L(z, 2b) = \bigcup_{a \in \ZZ} [az + 2b, (a + 1)z - 2b)$,
and $\overline{L(z, 2b)^d}$ is the complement (with respect to $\mathbb{R}^d$) of $L(z, 2b)^d$.
Therefore,
\begin{equation}
    \label{eqn:to_shift}
    \left| \bigcup_{u \in f(\mathbb{R}^d)} M^* \cap (U_\tau^\infty(f^{-1}(u)) - v) \right|
    \leq \left| M^* \cap (\overline{L(z, 2b)^d} - v) \right|.
\end{equation}
Hence, it suffices to show there exists $v^* \in \mathcal{V}^d$,
such that $|M^* \cap (\overline{L(z, 2b)^d} - v^*)| \leq \epsilon |M^*|$.
This $v^*$ may be viewed as a ``shift'' of $\overline{L(z, 2b)^d}$,
and we would use an averaging argument to show the existence of it.
For $j \in \mathcal{V}$, let $\hat{L}(j) := \bigcup_{a \in \ZZ}[az - 2b - j, az + 2b - j)$.
Notice that $\hat{L}(j)$'s are of small length which is $4b$,
and $\{ \hat{L}(j) : j \in \mathcal{V} \}$ is a partition of $\mathbb{R}$.
Moreover, it can be observed that $\hat{L}(0) = \overline{L(z, 2b)}$.

\paragraph{Picking $v^*$.}
For $i \in [d]$, define multiset $\proj_i(M^*) := \{ x_i : x = (x_1, \ldots, x_d) \in M^* \}$ as taking the $i$-th coordinates of points in $M^*$.
Now, by an averaging argument (over $j \in \mathcal{V}$),
for every $i \in [d]$, there exists $j_i^*$ such that 
$ |\proj_i(M^*) \cap \hat{L}(j_i^*)| \leq \frac{1}{|\mathcal{V}|} \cdot |M^*| $.
We pick $v^* := (j_1^*, \ldots, j_d^*)$.

\paragraph{Concluding error bound.}
We continue to analyze \eqref{eqn:to_shift} when plugging in $v = v^*$.
\begin{align*}
    |M^* \cap (\overline{L(z, 2b)^d} - v^*)|
    &\leq \sum_{i = 1}^d |\proj_i(M^*) \cap (\hat{L}(0) - j_i^*)| \\
    &= \sum_{i = 1}^d |\proj_i(M^*) \cap \hat{L}(j_i^*)| \\
    &\leq \frac{d}{|\mathcal{V}|} |M^*|
    = \frac{d}{T} |M^*|.
\end{align*}
Combining this with \eqref{eqn:DvtoMstar}, \eqref{eqn:remove_v} and \eqref{eqn:back_to_union},
we conclude that
\begin{align*}
    |\widehat{M}| = \min_{v \in \mathcal{V}^d} |D_v|
    \leq |D_{v^*}| \leq |M^*| + \frac{d(d + 1)}{T} |M^*|
    \leq (1 + \epsilon) |M^*|.
\end{align*}
This finishes the proof of \Cref{thm:mds}.
\qed

     \section{MPC Algorithms for RS in High Dimension}
\label{sec:ruling_set}

\begin{lemma}
    \label{thm:ruling_set}
    There is an MPC algorithm that given threshold $\tau>0$,  $0<\varepsilon<1$ and a dataset $P$ of $n$ points in $\mathbb{R}^{d}$ distributed across MPC machines with local memory $s\geq \poly(d\log n)$, with probability at least $1-1/n$ computes a $(\tau, O(\varepsilon^{-1}\frac{\log n}{\log\log n})\tau)$-ruling set for $P$ in $O(\log_{s} n)$ rounds, using total memory $O(n^{1+\varepsilon}\poly(d\log n))$.    
\end{lemma}

Our algorithm may be viewed as a Euclidean version of one-round Luby's,
with two major differences. One is an additional preprocessing algorithm \Cref{alg:rs_hash},
and this step is crucially needed to improve the ruling set parameter from $O(\log n)$ (which is what one-round Luby's algorithm can achieve in general~\cite{dist_graph_book}) to $(\log n / \log\log n)$.
The other is that it is not immediate to exactly simulate the one-round Luby's in high dimensional Euclidean spaces,
as the neighborhood around a data point can be huge and is not easy to handle in MPC.
To this end, we need to use some approximate ball $A^\beta_P(p, \tau)$ that is ``sandwiched'' between $B_P(p, \tau)$ and $B
_{P}(p, \beta \tau)$ for some $\beta$ and every point $p \in P$,
and an MPC procedure for this has been suggested in~\cite{CGJKV24} (restated in \Cref{lemma:geometric_aggregation}).
We would address this inaccuracy introduced by $A^\beta_P(\cdot, \cdot)$ in the analysis of the offline algorithm.
In the remainder of this section, we use the notations from \Cref{alg:rs_hash,alg:local_kcenter}.

\begin{algorithm}
    \begin{algorithmic}[1]
        \caption{Preprocessing, with input $P \subseteq \mathbb{R}^d$ of $n$ points, $\beta \geq 1, \tau > 0$}
        \label{alg:rs_hash}
\State let $\hash$  be a consistent hashing with parameter $\Gamma \gets O(d / \log\log n)$
        $\Lambda \gets \poly(d\log n)$ and diameter $\ell \gets O(\beta\Gamma \tau)$, using \Cref{lemma:consist_hash}
\State for every $z \in \hash(P)$, pick an arbitrary representative point $x \in \hash^{-1}(z) \cap P$,
        denoted as $\rep(z)$
        \State return $P' \gets \rep(\hash(P))$
\end{algorithmic}
\end{algorithm}

\begin{algorithm}
    \caption{Local algorithm for RS on $P'$ resultant from \Cref{alg:rs_hash}, same $\beta \geq 1, \tau > 0$}
    \label{alg:local_kcenter}
    \begin{algorithmic}[1]
\State for each $p\in P'$, pick a uniformly random label $h(p)\in [0,1]$
        
        \State initialize $R\gets \emptyset$, and for every $p \in P'$, $R \gets R \cup \{p\}$ if $p$ has the smallest label in $A_{P'}^{\beta}(p, \tau)$ 
            \label{line:rule}
        \State return $R$
    \end{algorithmic}
\end{algorithm}
    \begin{lemma}
        \label{lemma:C_independent_set}
$R$ is a $\tau$-independent set for $P'$ (with probability $1$).  
    \end{lemma}
\begin{proof} 
        Fix a point $x\in R$. 
        Suppose for the contrary that there is a point $y \neq x \in R$ satisfying $\dist(x, y)\leq \tau$. 
        Then we have that $y\in B(x, \tau)$ and $x\in B(y,\tau)$.  
        Observe that $R\subseteq P'$ which implies that $x, y\in P'$. 
        Since for every $p\in P'$, $B_{P'}(p, \tau)\subseteq A_{P'}^{\beta}(p, \tau)$, then we have that $y\in B_{P'}(x, \tau)\subseteq A_{P'}^{\beta}(x, \tau)$ and $x\in B_{P'}(y, \tau)\subseteq A_{P'}^{\beta}(y, \tau)$.  
        However, we have $h(x) = \min(h(A_{P'}^{\beta}(x, \tau)))$ and $h(y)=\min(h(A_{P'}^{\beta}(y, \tau)))$ by line~\ref{line:rule} of \Cref{alg:local_kcenter},
        which implies that $h(x)<h(y)$ and $h(y)<h(x)$, respectively, leading to a contradiction. 
        Therefore, we complete the proof. 
    \end{proof}

    \begin{lemma}
        \label{lemma:rs_Pprime_to_P}
        For any $\alpha \geq 1$,
        any $(\tau, \alpha \tau)$-ruling set for $P'$
        is an $(\tau, (\alpha + \beta\Gamma) \tau)$-ruling set for $P$.
    \end{lemma}
    \begin{proof}
        Let $S$ be a $(\tau, \alpha \tau)$-ruling set for $P'$.
        Since $S \subseteq P' \subseteq P$ and $S$ is $\tau$-independent by definition,
        it remains to verify for every $p \in P$, $\dist(p, S) \leq (\alpha + \ell) \tau$.
        Fix some $p \in P$.
        For a generic point set $W \subseteq \mathbb{R}^d$ and a point $x \in \mathbb{R}^d$,
        let $W(x)$ be the nearest point in $W$ to $x$.
        Then
        \begin{align*}
            \dist(p, S)
            \leq \dist(p, P'(p)) + \dist(P'(p), S(p))
            \leq \dist(p, \rep(\hash(p))) + \alpha \tau
            \leq \ell + \alpha \tau \leq (\alpha + \beta\Gamma) \tau.
        \end{align*}
    \end{proof}

    \begin{fact}
        \label{fact:Lambda_ub}
        For every $p \in P'$, $|A_{P'}^\beta(x, \tau)| \leq \Lambda$.
    \end{fact}
    \begin{proof}
        Follows immediately from the definition of parameters $\Gamma, \ell, \Lambda$ and \Cref{lemma:consist_hash}.
    \end{proof}
    
    \begin{lemma}
        \label{lemma:approximation_ratio}
For every $t \geq \Omega(\log^2 \Lambda)$,
        the set $R$ returned by \Cref{alg:local_kcenter} satisfies
        \begin{equation*}
            \forall p\in P',\quad \dist(p, R) \leq O(\beta t) \tau
        \end{equation*}
        with probability at least $1 - n \cdot \exp(-\Omega(t) \log(t / \log \Lambda))$.
    \end{lemma}
\begin{proof}
        It suffices to show that for every point $p\in P'$, with probability $1-\exp(-\Omega(t) \log(t / \log \Lambda))$,
        $\dist(p, R)\leq O(\beta t)\tau$, 
        since one can conclude the proof using a union bound for all points $p\in P'$.

        Now, fix a point $p\in P'$ and consider $\dist(p, R)$. 
        In order to analyze $\dist(p, R)$, we construct a sequence $S:= (x_{0}, x_{1}, \ldots, x_{T})$ using an auxiliary algorithm, stated in \Cref{alg:find_sequence},
        where the sequence $S$ starts at the point $x_{0}:= p$ and denote the length of $S$ as $T\geq 0$ which is random. 
        We emphasize that \Cref{alg:find_sequence} is defined with respect to the internal states of a specific (random) run of \Cref{alg:local_kcenter},
        and it does not introduce new randomness.
    \begin{algorithm}[H]
        \caption{Finding an assignment sequence $S=(x_{0}=p, \ldots, x_{T})$, for a given $p\in  P'$}
        \label{alg:find_sequence}
        \begin{algorithmic}[1]

            \State let $i\gets 0$, $x_{0}\gets p$, $S\gets (x_{0})$
            
            \While{\Cref{alg:local_kcenter} does not add $x_{i}$ at line~\ref{line:rule}}\label{line:while}
    
            \State let $x_{i+1}$ be the point from $A_{P'}^{\beta}(x_{i}, \tau)$ with the smallest label 
            \label{line:x_i+1}
    
            \State let $S\gets S \circ x_{i+1}$ and $i\gets i+1$
    
            \EndWhile
            \State return $S$
        \end{algorithmic}
    \end{algorithm}
Observe that in \Cref{alg:find_sequence}, for every $i\geq 0$, $x_{i+1}\in A_{P'}^{\beta}(x_{i}, \tau)$.  
        Since for every $p\in P'$, $A_{P'}^{\beta}(p, \tau)\subseteq B_{P'}(p, \beta \tau)$, then we have that
        \begin{equation*}
            \dist(p, R)\leq \sum_{i=1}^{T} \dist(x_{i-1}, x_{i}) \leq T\beta \tau
        \end{equation*}
        by triangle inequality.
        Hence, it remains to give an upper bound for $T$, and we show this in the following \Cref{lemma:bound_T} which is the main technical lemma.
        \begin{lemma}
            \label{lemma:bound_T}
For $t \geq \Omega(\log^2 \Lambda)$,
            we have $\Pr[T\geq t]\leq \exp(-\Omega(t) \log(t / \log \Lambda))$. 
        \end{lemma} 
        \begin{proof}
            The proof is postponed in \Cref{sec:proof_bound_T}.
        \end{proof}
        Finally, as mentioned, applying \Cref{lemma:bound_T} with a union bound finishes the proof of \Cref{lemma:approximation_ratio}.
        \end{proof}

\begin{proof}[Proof of \Cref{thm:ruling_set}]
        Our MPC algorithm for \Cref{thm:ruling_set} is obtained via an MPC implementation of the offline \Cref{alg:rs_hash,alg:local_kcenter}. 
        However, before we do so, our algorithm requires to run the
        Johnson-Lindenstrauss (JL) transform \cite{JL84} on the dataset as preprocessing to reduce $d$ to $d = O(\log n)$, using parameter $\epsilon = O(1)$, restated as follows with our notations.
        \begin{lemma}[JL transform~\cite{JL84}]
            For some $0 < \epsilon < 1$,
            let $M \in \mathbb{R}^{d' \times d}$ be a random matrix such that every entry is an independent standard Gaussian variable $N(0, 1)$ where $d' = O(\epsilon^{-2}\log n )$,
            then the mapping $g : x \mapsto 1 / \sqrt{d'} \cdot Mx$ satisfies that
            with probability $1 - 1 / \poly(n)$,
            $\forall x, y \in P$, $\dist(g(x), g(y))  \in (1 \pm \epsilon) \dist(x, y)$.
        \end{lemma}
        This JL transform only needs $O(1)$ rounds to run in MPC, since one can generate the matrix $M$ in a lead machine, which uses space $O(d \log n)$, and then broadcast to every other machines.
        Since the pairwise distance between data points is preserved, and that our algorithms only use the distance between data points,
        it is immediate that any $(\tau, \alpha \tau)$-ruling set on $g(P)$
        is as well a $(\Theta(\tau), \Theta(\alpha \tau))$-ruling set on $P$.
        Hence, it suffices to work on $g(P)$ which is of dimension $d' = O(\log n)$.
        Without loss of generality, we can simply assume $P$ is of dimension $O(\log n)$.

        Now, we turn to implement \Cref{alg:rs_hash,alg:local_kcenter} in MPC.
        For \Cref{alg:rs_hash}, since the hash in \Cref{lemma:consist_hash} that we use is data oblivious and only requires $\poly(d) = \poly(\log n)$ space,
        all machines can generate the same hash without communication.
        The other steps in \Cref{alg:rs_hash} may be implemented using standard MPC procedures including sorting, broad-casting and converge-casting.
        For \Cref{alg:local_kcenter},
        the the only nontrivial step is to implement $A_{P'}^{\beta}(\cdot, \cdot)$.
        To this end, we make use of the following lemma from~\cite{CGJKV24}.
In particular, our MPC implementation of \Cref{alg:local_kcenter} applies \Cref{lemma:geometric_aggregation} with $\beta = O(\varepsilon^{-1})$.
        This finishes the description of the algorithm for \Cref{thm:ruling_set}.
\begin{lemma}[{\cite[Theorem 3.1]{CGJKV24}}]
    \label{lemma:geometric_aggregation}
    There is a deterministic MPC algorithm that takes as input
    $0 < \epsilon < 1$, $\tau\geq 0$, $P\subset \mathbb{R}^{d}$ of $n$ points and for each $p \in P$ a value $h(p) \in \mathbb{R}$,
    distributed across machines with local memory $s\geq \poly(d\log n)$,
    computes for every $p\in P$ a value $\min(h(A_{P}(p, \tau)))$, where $A_{P}(p, \tau)$ is an arbitrary set that satisfies 
    \begin{equation}
        B_{P}(p, \tau)\subseteq A_{P}(p, \tau)\subseteq B_{P}(p,  O(\epsilon^{-1} \tau )),
    \end{equation}
    in $O(\log_{s} n)$ rounds and $O(n^{1 + \epsilon}\poly(d\log n))$ total memory.
\end{lemma}

        Now we turn to the analysis.
        Observe that the assumptions regarding $A_{P'}(\cdot, \cdot)$
        underlying the analysis of \Cref{alg:rs_hash,alg:local_kcenter} remain valid.
        By \Cref{lemma:C_independent_set}, the found set $R$ is also an $\tau$-independent for $P$ as $R \subseteq P' \subseteq P$.
        Moreover, since we assume $d = O(\log n)$, the parameters $\Gamma = \log n / \log\log n$,
        and $\Lambda = \poly(\log n)$.
        Therefore, applying \Cref{lemma:approximation_ratio} with $t = O(\log n / \log\log n)$,
        we conclude that $\forall p \in P'$, $\dist(p, R) \leq O(\epsilon^{-1} \log n / \log\log n) \tau$,
        with probability $1 - 1 / \poly(n)$.
        Combining with \Cref{lemma:rs_Pprime_to_P},
        we conclude that $R$ is $(\tau, O(\epsilon^{-1}\log n / \log\log n)\tau)$-ruling set for $P$,
        with probability $1 - 1 / \poly(n)$.

Finally, for the round complexity, local memory and total memory, these are dominated by the parallel invocations of \Cref{lemma:geometric_aggregation}, which are $O(\log_{s} n)$, $\poly(d\log n)$ and $O(n^{1+\varepsilon}\poly(d\log n))$, respectively.  
        Therefore, we complete the proof of \Cref{thm:ruling_set}. 
\end{proof}

\subsection{Proof of Lemma~\ref{lemma:bound_T}: Upper Bound The Length of Auxiliary Sequence}
 \label{sec:proof_bound_T}

        Let $S$ be the (random) sequence returned by \Cref{alg:find_sequence},
        and recall that we write $S =(x_0 = p, \ldots, x_T)$.
        For two sequences $S'$ and $S''$, let $S'\circ S''$ denote their concatenation; when $S'' = (x) $ is a singleton sequence, we write $S'\circ x$, which is a shorthand for appending the point $x$ to the sequence $S'$.     
        In addition, let $S'\sqsubseteq S''$ denote that $S'$ is a prefix of $S''$. 
        
        Let $\mathcal{S}$ be the set of all sequences $S':=(x_{0}', x_{1}', \ldots )$ such that: (1) $x_{0}' = p$; (2) for every $i\geq 0$, $x_{i+1}\in A_{P'}^{\beta}(x_{i}, \tau)$.  
        Consider some $S'\in \mathcal{S}$ denoted as $S' := (x_{0}', x_{1}', \ldots, x_{m}')$ and $m\geq 1$.  
        Define 
        \begin{align}
            \label{eqn:Acond_Anew}
            \Acond(S')&:= \bigcup_{i=0}^{m-1}A_{P'}^{\beta}(x_{i}', \tau)\nonumber\\ 
            \Anew(S')&:= A_{P'}^{\beta}(x_{m}', \tau)\setminus \Acond(S'). 
        \end{align} 
        Let $\Econd(S')$ denote the event that $S'\sqsubseteq S$, i.e., \Cref{alg:find_sequence} gets to line~\ref{line:while} with the current sequence equal to $S'$. 
        Let $\Eextend(S')$ denote the event that $\Econd(S')$ holds and in iteration $i=m$ of \Cref{alg:find_sequence},  $x'_{m}$ passes the test in line~\ref{line:while} and the point picked in line~\ref{line:x_i+1} is from $A_{P'}^{\beta}(x_{m}', \tau)$.  
        The next lemma establishes an upper bound for the probability of the sequence gets a new element appended,
        with respect to  $\frac{|\Anew(S')|}{|\Anew(S')\cup \Acond(S')|}$
        which may be understood as the local geometric growth of the dataset.  

        \begin{lemma}
            \label{lemma:ext_cond}
            For every $S'\in \mathcal{S}$, we have that 
            \begin{equation*}
                \label{eqn:formalized}
                \Pr[\Eextend(S')\mid \Econd(S')]\leq \frac{|\Anew(S')|}{|\Anew(S')\cup \Acond(S')|} .
            \end{equation*} 
        \end{lemma}
        \begin{proof} 
            Let $S' := (x_{0}'=p, x_{1}', \ldots, x_{m}')\in \mathcal{S}$ be a fixed sequence with a length of $m+1$.  
            Observe that the event $\Econd(S')$ may depend on the random label of points in the set $\Acond(S') = \bigcup_{i=0}^{m-1} A_{P'}^{\beta}(x_{i}', \tau)$, but not on those of points outside $\Acond(S')$. 
            The label choices for points outside $\Acond(S')$ and inside $\Acond(S')$ are independent, hence by the principle of deferred decisions, conditioning on $\Econd(S')$ does not change the distribution of label choices for points in $\Anew(S')=A_{P'}^{\beta}(x_{m}', \tau)\setminus \Acond(S')$, which in turn affect the event $\Eextend(S')$ (e.g., whether $x_{m}'$ passes the test). 
            
            When $\Econd(S')$ occurs, each point $x_{i}'$, for $i\in[m]$, has the smallest label in $A_{P'}^{\beta}(x_{i-1}', \tau)$,
            and it follows by induction that $x_{m}'$ has the smallest label in $\Acond(S')$. 
            When $\Eextend(S')$ occurs, some $x_{m+1}'\neq x_{m}'$ has the smallest label in $A_{P'}^{\beta}(x_{m}', \tau)$, and thus $h(x_{m+1}')< h(x_{m}') = \min(h(\Acond(S')))$, implying that $x_{m+1}'\notin \Acond(S')$ and in fact $x_{m+1}'\in \Anew(S')$. 
            Therefore, 
        \begin{eqnarray}
            \label{eqn:condition_probability}
            \Pr[\Eextend(S') \mid \Econd(S')] \leq \sum_{q\in \Anew(S')} \Pr[h(q) = \min(h(A_{P'}^{\beta}(x_{m}', \tau)))\mid \Econd(S')].
        \end{eqnarray} 
    
        \begin{claim}[\cite{CGJKV24}, Claim 4.13]
            \label{claim:partial_order}
            Fix a finite domain $\mathcal{U}$, and let $g: \mathcal{U}\to [0, 1]$ be random, such that each $g(a)$ is chosen independently and uniformly from $[0, 1]$. 
            Then for every subset $X\subseteq \mathcal{U}$, element $a\in X$, and partial order $\mathcal{P}$ on $X\setminus\{a\}$, we have 
            \begin{equation*}
                \Pr_{g} [g(a) = \min(g(X)) \mid g\propto \mathcal{P}] = 1/|X|, 
            \end{equation*}
            where $g\propto \mathcal{P}$ denotes consistency in the sense that $g(x)\leq g(y)$ whenever $x\prec_{\mathcal{P}} y$. 
        \end{claim} 

        To upper bound \eqref{eqn:condition_probability}, it suffices to give for every $q \in \Anew(S')$ an upper bound for $\Pr[h(q) = \min(h(A_{P'}^{\beta} (x_{m}', \tau)))\mid \Econd(S')]$. 
        Notice the event $\Econd(S')$ defines a partial order via the label $h$,
        such that for $i\in[m]$, each point $x_{i}'$ has the smallest label in $A_{P'}^{\beta}(x_{i-1}', \tau)$,
        i.e., for each $y\neq x_{i}'\in A_{P'}^{\beta}(x_{i-1}', \tau)$, $h(x_{i}')< h(y)$, and this is in particular a partial order on $\Anew(S')\cup \Acond(S') \setminus \{q\}$. 
        Now, apply \Cref{claim:partial_order} with $g:= h$, $a:= q$, $X:= \Anew(S')\cup \Acond(S')$ which contains $a$, and a partial order $\mathcal{P}$ defined from the event $\Econd(S')$ (as mentioned above). 
            Hence,  
        \begin{equation*}
            \forall q\in \Anew(S'), \quad \Pr[h(q) = \min(h(A_{P'}^{\beta}(x_{m}', \tau))) \mid \Econd(S')] = \frac{1}{|\Anew(S')\cup \Acond(S') |}.
        \end{equation*}
            Plugging this into \eqref{eqn:condition_probability}, we obtain that $\Pr[\Eextend(S')\mid \Econd(S')]\leq \frac{|\Anew(S')|}{|\Anew(S')\cup \Acond(S')|}$, 
            and this completes the proof of \Cref{lemma:ext_cond}.  
        \end{proof}

        We continue the proof of \Cref{lemma:bound_T}.
Let $f: \mathcal{S}\to \mathbb{Z}$ be a function such that for each $S'\in \mathcal{S}$,  
        \begin{align*}
            f(S') := \begin{cases}
               i & |\Anew(S')| \in [2^{i-1}, 2^{i})  \\
               -\infty &  \Anew(S') = \emptyset.
            \end{cases}
        \end{align*}
        The following fact follows from \Cref{fact:Lambda_ub}.
        \begin{fact}
            For any $S' \in \mathcal{S}$, we have $f(S') \in \{-\infty\} \cup [\lceil \log_2 \Lambda \rceil] $.
        \end{fact}
Intuitively, this function $f$ takes value $i$ when $\Anew(S')$ is around $[2^{i-1}, 2^{i})$. 
Now, for a given sequence $S'$, we consider the \emph{configuration} of $S'$ as the $f$ values for every prefix of $S'$.
Namely, for every $S'\in\mathcal{S}$ denoted as $S':=(x_{0}', x_{1}', \ldots, x_{m}')$ and $i\geq 0$, let $S'_{[0, i]}$ denote the prefix of $S'$ ending with $x_{i}'$, 
        and denote $\conf(S'):= (f(S'_{[0,1]}), \ldots, f(S'_{[0, m]}))$ as the \emph{configuration} of $S'$.
        For $\pi\in (\{-\infty\} \cup [\lceil \log_2\Lambda \rceil])^{i}$ (for some $i\geq 1$), let $\mathcal{S}_{\pi} := \{S'\in \mathcal{S}: \conf(S') = \pi\}$. 

    The next lemma is a key lemma which upper bounds the probability that a sequence has length at least $t$, for sequences of a fixed configuration $\pi$.
    This lemma would be used to conclude \Cref{lemma:bound_T} by a union bound over all possible configurations $\pi \in (\{-\infty\} \cup [\lceil\log_2 \Lambda\rceil])^t$.

\begin{lemma}
    \label{lemma:key}
    For every  $\pi\in (\{-\infty\} \cup [\lceil \log_2 \Lambda \rceil])^{t}$, $\sum_{S'\in \mathcal{S}_{\pi}} \Pr[S'\sqsubseteq S]\leq \exp(- \Omega(t\log (t / \log \Lambda)))$. 
\end{lemma}
\begin{proof} 
    For every $i\in [t]$, let $\pi_{\leq i}$ denote the prefix of $\pi$ with a length of $i$, and let $\pi_{i}$ denote the $i$th item of $\pi$.  
    Observe that if $\pi_i = -\infty$ for $i < t$, then $\Pr[S' \sqsubset S] = 0$ for any $S' \in \mathcal{S}_\pi$.
Hence, we only need to verify the claim for a $\pi\in (\{-\infty\} \cup [\lceil \log_2 \Lambda \rceil])^{t}$ such that $\pi_i \neq -\infty$ for every $i < t$.

    Now, for $i  \geq 2$, we start with proving 
    \begin{equation}
    \label{eqn:Si_Si-1}
        \sum_{S'\in\Spi{\leq i}} \Pr[S'\sqsubseteq S] \leq \sum_{S''\in\Spi{\leq i-1}} \Pr[\exists x\in P', S''\circ x\sqsubseteq S].
    \end{equation} 
    For every $i\in [t]$, let $\phi_i: \Spi{\leq i} \to \{(S'', x): S''\in\Spi{\leq i-1}, x\in P'\}$ be a mapping,
    where for $S'\in \Spi{\leq i}$, $\phi_i(S')$ maps to the unique pair $(S'', x)$ satisfying $S' = S''\circ x$ (i.e., $S'' = S'_{[0, i - 1]}$).   
    Then for every $S' \in \Spi{\leq i}$, $\Pr[S' \sqsubseteq S] = \Pr[S'' \circ x \sqsubseteq S]$
    where $(S'', x) = \phi_{i}(S')$.
    Moreover, $\phi_i$ is an injection by definition.
    Therefore,
    \begin{align*}
        \sum_{S'\in\Spi{\leq i}} \Pr[S'\sqsubseteq S]
        &= \sum_{\substack{ (S'', x) :\\
        (S'', x) = \phi_i(S'),S' \in \Spi{\leq i} } } \Pr[S'' \circ x \sqsubseteq S] \\
        &\leq \sum_{S''\in\Spi{\leq i-1}} \sum_{x \in P'} \Pr[S''\circ x\sqsubseteq S] \\
        &= \sum_{S''\in\Spi{\leq i-1}} \Pr[\exists x\in P', S''\circ x\sqsubseteq S],
    \end{align*}
    where the inequality uses that $\phi_i$ is injection.
    This concludes \eqref{eqn:Si_Si-1}.
    
    Therefore, we have that 
    \begin{align}
        \sum_{S'\in \Spi{\leq i}}\Pr[S'\sqsubseteq S] 
        &\leq  \sum_{S''\in \Spi{\leq i-1}} \Pr[\exists x\in P', S''\circ x\sqsubseteq S]\nonumber\\
        &=\sum_{S''\in \Spi{\leq i-1}} \Pr[\exists x\in P', S''\circ x\sqsubseteq S\mid S''\sqsubseteq S] \cdot \Pr[S''\sqsubseteq S]\nonumber\\
        &= \sum_{S''\in \Spi{\leq i-1}} \Pr[\Eextend(S'')\mid \Econd(S'')]\cdot \Pr[S''\sqsubseteq S]
\nonumber\\
        &\leq \sum_{S''\in \Spi{\leq i-1}} \frac{|\Anew(S'')|}{|\Anew(S'')\cup \Acond(S'')|} \cdot \Pr[S''\sqsubseteq S] \nonumber \\
        &\leq \sum_{S''\in \Spi{\leq i-1}} \frac{2^{1 + \pi_{i - 1}}}{\sum_{j=1}^{i - 1} 2^{\pi_j}} \cdot \Pr[S''\sqsubseteq S] \label{eqn:pi_recursion}
    \end{align}
    where the second inequality follows from \Cref{lemma:ext_cond},
    and the last inequality follows from \eqref{eqn:Acond_Anew}.
    The assumption that $\pi_i = -\infty$ may happen only when $i = t$
    ensures that the denominators $\sum_{j = 1}^{i - 1}2^{\pi_j}$ cannot be zero.
    Applying \eqref{eqn:pi_recursion} inductively,
    \begin{align*}
        \sum_{S'\in \Spi{\leq t}}\Pr[S'\sqsubseteq S] 
        \leq \prod_{i = 1}^{t - 1} \frac{2^{1 + \pi_{i}}}{\sum_{j=1}^{i} 2^{\pi_j}}.
    \end{align*}
    For every $i \in [\lceil \log_2 \Lambda \rceil]$,
    let $G_i := \{ \pi_j = i : j \in [t - 1] \} $ and let $w_i := |G_i|$.
    Then 
    \begin{align*}
        \prod_{i = 1}^{t - 1} \frac{2^{1 + \pi_{i}}}{\sum_{j=1}^{i} 2^{\pi_j}}
        &\leq \prod_{i \in [\lceil \log_2 \Lambda \rceil]} \frac{ 2^{(1 + i)w_i} }{2^{i w_i} \cdot w_i!}
        \leq \prod_{i \in [\lceil \log_2 \Lambda \rceil]} \frac{ 2^{w_i} }{w_i!}
&\leq \prod_{i \in [\lceil \log_2 \Lambda \rceil]} O(\sqrt w_i) \cdot \frac{ 2^{w_i} }{(w_i / e)^{w_i}},
    \end{align*}
    where the first inequality follows from the fact that
    \begin{align*}
        \prod_{i = 1}^{t - 1} \sum_{j = 1}^i 2^{\pi_j}
        \geq \prod_{i \in [\lceil \log_2 \lambda \rceil]} \prod_{j  \in G_i} \sum_{s = 1}^j 2^{\pi_s}
        \geq \prod_{i \in [\lceil \log_2 \lambda \rceil]} \prod_{j  \in G_i} \sum_{s \in [j] : s \in G_i} 2^{\pi_s}
        &= \prod_{i \in [\lceil \log_2 \lambda \rceil]} \prod_{j  \in [w_i]} j 2^{i} \\
        &\geq \prod_{i \in [\lceil \log_2 \lambda \rceil]}  2^{i w_i} w_i!
    \end{align*}
    Taking logarithm,
    \begin{align*}
        \log 
        \prod_{i = 1}^{t - 1} \frac{2^{1 + \pi_{i}}}{\sum_{j=1}^{i} 2^{\pi_j}}
        &\leq \log \prod_{i \in [\lceil \log_2 \Lambda \rceil]} O(\sqrt w_i) \cdot \frac{ 2^{w_i} }{(w_i / e)^{w_i}} \\
        &= \sum_{i \in [\lceil \log_2 \Lambda \rceil]} \log\left( O(\sqrt w_i) \cdot \frac{2^{w_i}}{(w_i / e)^{w_i}} \right) \\
        &\leq \sum_{i \in [\lceil \log_2 \Lambda \rceil]} O( \log(w_i) + w_i - w_i\log(w_i / e) ) \\
        &\leq \sum_{i \in [\lceil \log_2 \Lambda \rceil]} O( - w_i\log(w_i) ) \\
        &\leq O(-t \log(t / \log \Lambda)),
    \end{align*}
    where the last inequality follows from the fact that $\sum_i w_i \leq t$ and Jensen's inequality.
    We conclude that $\log \sum_{S'\in \Spi{\leq t}}\Pr[S'\sqsubseteq S] \leq O(-t \log(t / \log \Lambda))$,
    and this finishes the proof.
\end{proof}

\begin{proof}[Proof of \Cref{lemma:bound_T}]
Finally, we conclude \Cref{lemma:bound_T} by a union-bound argument.
\begin{align*}
    \Pr[T\geq t]
    &= \sum_{S'\in \mathcal{S}:|S'| = t+1} \Pr[S'\sqsubseteq S] \\
    &= \sum_{\pi \in [\lceil \log_2 \Lambda \rceil]^{t}}\sum_{S'\in \mathcal{S}_{\pi}} \Pr[S'\sqsubseteq S] \\
    &\leq (\log\Lambda )^{O(t)} \cdot \exp(-\Omega(t \log(t / \log\Lambda))) \\
    &\leq \exp(O(t)\log\log \Lambda - \Omega(t) \log(t / \log \Lambda)) \\
    &\leq \exp(-\Omega(t) \log(t / \log \Lambda)),
\end{align*}
where the first inequality is by applying \Cref{lemma:key},
and the last inequality uses that $t \geq \Omega(\log^2 \Lambda)$.
This completes the proof of \Cref{lemma:bound_T}.
\end{proof}

     \section{Proof of \Cref{thm:low_dim_2approx,thm:low_dim_bicrit,thm:kcenter_high_dim}: MPC Algorithms for \kCenter}
\label{sec:kcenter}

We present the proofs for \Cref{thm:low_dim_2approx,thm:low_dim_bicrit,thm:kcenter_high_dim}. 
As mentioned, we obtain these results by reducing to geometric versions of either RS or MDS.
Hence, we present the reduction as a meta algorithm, as in \Cref{alg:kcenter}, instead of separated concrete algorithms.
The concrete algorithms can be derived by plugging in suitable combination of RS or MDS results.

We note that this meta algorithm only finds a set of \emph{centers} as approximate solution,
and we discuss in \Cref{sec:assignment} how to also find the approximate \emph{assignment} from data points to the found center set.
This procedure introduces a slight loss in the ratio, but it is only minor;
in particular, in low dimension the ratio increases by only a $(1 + \epsilon)$ factor,
and in high dimension the ratio increases by a constant factor.

\begin{algorithm}[ht]
    \caption{MPC Algorithms for \kCenter via RS and MDS with parameter $0 < \epsilon < 1$}
    \label{alg:kcenter}
    \begin{algorithmic}[1]
        \State return $P$ as the solution if $P$ contains at most $k$ distinct points 
        \label{line:test_distinct}
        \State compute an $\alpha:=\poly(n)$-approximate value $\APX$ for \kCenter on $P$ using \Cref{lemma:guess_OPT}
        \label{line:guess}
        \State let $Z\gets \{i\in \ZZ: \APX / \alpha \leq (1+\varepsilon)^{i} \leq \APX \}$ be a set of integer power of $(1 + \varepsilon)$
        \label{line:candidate_value}
        
        \For{$\tau \in Z$ in parallel}

        \State (RS-based approx.) 
        \label{line:ruling_set}
find $I_{\tau}\gets$ $(2\tau, \gamma\tau)$-ruling set for $P$ by \Cref{thm:MIS} or \Cref{thm:ruling_set} 

        \Comment{where $\gamma = 2(1 + \epsilon)$ as in \Cref{thm:MIS}, or $\gamma = O(\epsilon^{-1}\frac{\log n}{\log\log n})$ as in \Cref{thm:ruling_set}}

        \State (MDS-based approx.) 
        
        Let $D_\tau$ be $(1 + \epsilon)\tau$-DS with size $|D_\tau| \leq (1 + \epsilon) |\MDS_\tau(P)|$   by \Cref{thm:mds}
        \label{line:bicsolution}

        \EndFor

        \State (RS-based approx.) 
        let $\tau^* \in Z$ be the smallest $\tau$ such that $|I_\tau| \leq k$, return $I_{\tau^*}$
        \label{line:2ereturn}

        \State (MDS-based approx.) let $\tau^* \in Z$ be the smallest $\tau$ such that $|D_\tau| \leq (1+\varepsilon)k$, return $D_{\tau^*}$
        \label{line:bicreturn}

    \end{algorithmic}
\end{algorithm}

\paragraph{Algorithm.}
The algorithm, listed in \Cref{alg:kcenter}, takes as input a dataset $P \subseteq \mathbb{R}^d$, integer $k \geq 1$ and a parameter $0 < \epsilon < 1$.
The algorithm starts with testing if $\OPT = 0$, which is equivalent to $P$ has at most $k$ distinct points.
If this does not happen, the algorithm ``guesses'' OPT (up to $(1 + \varepsilon)$ factor).
This step requires to first obtain a $\poly(n)$-approximate value $\APX$ for $\OPT$ using \Cref{lemma:guess_OPT} (which can be found in \Cref{subsec:guess}), and this is one of the randomness in the algorithm. 
Then we search for a suitable threshold $\tau$ and apply \Cref{thm:MIS,thm:ruling_set,thm:mds} to compute a center set.

\paragraph{Round complexity analysis.}
Line~\ref{line:test_distinct} can be implemented in $O(\log_s n)$ MPC rounds via sorting and broadcasting. 
The number of rounds for line~\ref{line:guess} is dominated by \Cref{lemma:guess_OPT}, which is $O(\log_{s} n)$ rounds. 
Observe that $|Z| = O(\log n)$, so there are $O(\log n)$ parallel invocations of the parallel-for.
Hence, the number of rounds for the parallel-for is of the same order as in \Cref{thm:MIS}, \Cref{thm:mds} and \Cref{thm:ruling_set}, which is $O(\log_s n)$. 
This finishes the round analysis.

\paragraph{Space complexity analysis.}
For the space, it is dominated by that of \Cref{thm:MIS}, \Cref{thm:mds} and \Cref{thm:ruling_set}, multiplied by $|Z|$ which is the number of parallel iterations (since they need to be replicated). 
By \Cref{thm:MIS,thm:mds},
we conclude that the local space requirement is $s \geq \Omega(\varepsilon^{-1}d)^{\Omega(d)} \poly(\log n)$,
and the total space is $O(n \poly( \log n))\cdot O(\varepsilon^{-1}d)^{O(d)}$ for $(2 + \epsilon)$-approximation and $(1 + \epsilon, 1 + \epsilon)$-approximation, which concludes the space analysis for both \Cref{thm:low_dim_2approx,thm:low_dim_bicrit}.
By \Cref{thm:ruling_set}, we conclude that the local space requirement is $s\geq \poly(d\log n)$, and the total space is $O(n^{1+\epsilon}\poly(d\log n))$ for $O(\epsilon^{-1}\log n/ \log \log n)$-approximation, which completes the space analysis of \Cref{thm:kcenter_high_dim}. 
This concludes the space analysis.

\paragraph{Error analysis: the common part.}
As we mentioned, if the algorithm terminates at line~\ref{line:test_distinct}, then $\OPT = 0$ and the entire dataset itself is the optimal solution.
Now assume the dataset consists of more than $k$ distinct points, which means $\OPT > 0$. 
As in line~\ref{line:guess} and~\ref{line:candidate_value}, we use \Cref{lemma:guess_OPT} to obtain $\poly(n)$-approximation to $\OPT$
and it succeeds with probability at least $1-1/n$.
Hence, with probability at least $1-1/n$, there exists some $\hat{\tau}\in Z$ such that $\OPT\leq \hat{\tau}\leq (1+\varepsilon)\OPT$. 
We condition on this high-probability event happens in the remainder of the proof.

\paragraph{Error analysis for RS-based approximation.}

We start with justifying that the algorithm is well-defined, i.e., $\tau^*$ in line~\ref{line:2ereturn} must exist.
In particular, we claim that if $\tau \geq \OPT(P)$ then $|I_\tau| \leq k$. 
This claim implies the existence of $\tau^*$ since $\tau = \hat{\tau}$ satisfies $\tau \geq \OPT$,
which also implies
\begin{equation}
    \label{eqn:taustar}
    \tau^* \leq \hat{\tau} \leq (1 + \epsilon) \OPT.
\end{equation}

Now, consider some $\tau\geq \OPT$ (noticing that $\hat{\tau}$ satisfies this). 
Then we make use of the following standard fact (proved for completeness in \Cref{fact:is_ub}, see also~\cite{DBLP:journals/jacm/HochbaumS86}),
and specifically apply it with  $S=P$, $\mu = \tau\geq \OPT$.
It readily implies $|I_\tau| \leq k$ since $I_{\tau}$ is a $2\tau$-independent set for $P$.  
    
\begin{restatable}{fact}{factisub}
\label{fact:is_ub}
Let $S \subseteq \RR^d$. For $\mu \geq \OPT(S)$, any subset of $S$ that is a $2\mu$-independent set for $S$ has at most $k$ points. 
\end{restatable}
\begin{proof}
Suppose for the contrary that there is some $2\mu$-independent set $\IS \subseteq S$ for $S$ with $|\IS| \ge k + 1$.
Consider a clustering $\mathcal{C} = \{C_1, C_2, \dots, C_k\}$ in an optimal solution for \kCenter on $S$.
Observe that each cluster $C_i \in \mathcal{C}$ has diameter $\diam(C_i) \leq 2\OPT(S) \leq 2 \mu$.
Now, consider any two distinct points $x \neq y \in \IS$. Since $\dist(x, y) > 2\mu$ by definition,
they cannot belong to the same cluster. Therefore, since $|\IS| \ge k+1$, there must be a point in \IS that does not belong to any cluster of $\mathcal{C}$, leading to a contradiction (since $\IS \subseteq S$ and $\bigcup_{i=1}^k C_i = S$).
This finishes the proof of \Cref{fact:is_ub}.
\end{proof}

Since $|I_{\tau^*}|\leq k$, $I_{\tau^*}$ is a feasible solution for \kCenter. 
It remains to analyze $\cost(P, I_{\tau^*})$. 
Since $I_{\tau^*}$ is a $(2\tau^{*}, \gamma\tau^{*})$-ruling set for $P$, we know that $\cost(P, I_{\tau^{*}})\leq \gamma\tau^{*}\leq \gamma(1+\varepsilon)\OPT$, where the last inequality follows from \eqref{eqn:taustar}.
This finishes the error analysis for the $(2 + \epsilon)$-approximation and $O(\epsilon^{-1}\log n/ \log \log n)$-approximation by substituting the corresponding values of $\gamma$.

\paragraph{Error analysis for MDS-based approximation.}
Similarly, we start with showing that the algorithm is well-defined, i.e., $\tau^*$ in line~\ref{line:bicreturn} must exist,
by showing that if $\tau \geq \OPT$ then $|D_\tau| \leq (1 + \epsilon) k$.
This further implies $\tau^* \leq \hat{\tau} \leq (1 + \epsilon) \OPT$ (and hence $\tau^*$ exists). 

To see this, let $C^*$ be the optimal solution to \kCenter on the input $P$. 
Then for every point $x\in P$, $\dist(x, C^*)\leq \OPT \leq \tau$.
Hence, $C^*$ is a $\tau$-DS for $P$. 
Therefore, 
\begin{equation*}
    |D_{\tau}|\leq (1+\varepsilon)|\MDS_{\tau}(P)|\leq (1+\varepsilon)|C^*|\leq (1+\varepsilon)k,
\end{equation*} 
where the first inequality is by \Cref{thm:mds} and the second inequality follow from the optimality of $\MDS_\tau$. 

Since the algorithm always returns a set of at most $(1 + \epsilon)k$ points, it remains to show $\cost(P, D_{\tau^*}) \leq (1 + O(\epsilon)) \OPT$.
Now, since $\tau^* \leq (1 + \epsilon) \OPT$, and by the fact that $D_{\tau^*}$ is a $(1 + \epsilon)\tau^*$-DS for $P$, we have
\begin{align*}
    \cost(P, D_{\tau^*})
    = \max_{x \in P} \dist(x, D_{\tau^*}) 
    \leq (1 + \epsilon) \tau^* 
    \leq (1 + O(\epsilon)) \OPT.
\end{align*}
This finishes the proof of \Cref{thm:low_dim_2approx,thm:low_dim_bicrit,thm:kcenter_high_dim}. 
\qed

\subsection{Coarse Approximation for Optimal Value of \kCenter}
\label{subsec:guess}
\begin{lemma}
    \label{lemma:guess_OPT}
    There is a randomized MPC algorithm, that given a dataset $P$ of $n$ points in $\mathbb{R}^{d}$ distributed across MPC machines with local memory $s\geq \Omega(\poly(d\log n))$,
    computes $E \geq 0$ such that $\OPT \leq E \leq O(n^7) \OPT$ for \kCenter
    with probability at least $1-1/n$,
    in $O(\log_{s} n)$ rounds using total memory $O(n \poly(d\log n))$. 
\end{lemma}

\begin{proof}
    
Let $v\in \mathbb{R}^{d}$ be an i.i.d. standard Gaussian vector, i.e., each entry of $v$ satisfies $v_{i}\sim \mathcal{N}(0, 1)$. 
For a point set $P\subset \mathbb{R}^{d}$,
we do the inner product to every $x \in P$ with respect to $v$,
and denote $P':= \{\langle v, x\rangle: \forall x\in P\}$ as this set.
This set can be viewed as the 1D projection of $P$.

We claim that, the event ``for every $x, y \in P$, $|\langle x, v \rangle - \langle y, v\rangle| \in [\Omega(1 / n^3), O(n^3)] \dist(x, y)$'' happens with probability $1 - 1 / n$.
To see this claim, fix some $x, y \in P$. Observe that $\frac{\langle x, v \rangle - \langle y, v \rangle}{\dist(x, y)} \sim \mathcal{N}(0, 1)$.
Plug this into the following standard fact about Gaussian, and do a union bound for all pairs $x, y$ yields the claim.
\begin{fact}
    \label{fact:dimension reduction}
    For some $z \geq 1$ and $u\sim\mathcal{N}(0, 1)$, we have that 
    $\Pr[|u| \leq 1 / z] \leq O(1 / z)$, and
    $\Pr[|u| \geq z] \leq O(1 / z)$.
\end{fact}
Hence, it suffices to approximate \kCenter within $O(n)$ factor on the 1D input $P'$ (where the distance is defined as the absolute value of the difference).
It would be convenient (and without loss of generality) to assume all points are distinct and there are at least $n \geq k + 1$ of them.

If $k = 1$, then we simply define the maximum minus the minimum as the estimation $E'$, and this is $2$-approximation.
Otherwise, for $k \geq 2$, we sort $P'$ with respect to their value/coordinate in 1D,
and we find the $k$-th largest gap between consecutive points, denoted as $r_k$, and we define $E' := n\cdot r_k$.

We claim that $E'$ satisfies $\Omega(\OPT(P')) \leq E' \leq O(n)\OPT(P')$. 
On one hand, $\OPT(P')$ must be at least $r_k/2$
for otherwise it requires at least $k + 1$ centers to cover the consecutive point pairs whose gap is at least $r_k$. 
This implies that $E'\leq O(n)\cdot \OPT(P')$. 

On the other hand, we show how to  construct a feasible solution whose cost is at most $O(n\cdot r_k)$,
and this would imply $\OPT(P')\leq \cost(P', C_{\mathcal{P}}) \leq O(n\cdot r_k) = O(E')$. 
Let $\mathcal{P} := \{ (p_i, q_i) \}_{i}$ be the consecutive point pairs (after the above sorting)
such that for each pair $(p_{i}, q_{i})\in \mathcal{P}$, $\dist(p_{i}, q_{i})>r_{k}$, and order $(p_i, q_i)$ by $p_i < q_i$ (in 1D coordinate). 
We have $|\mathcal{P}|\leq k-1$,
since otherwise,  there would be at least $k$ consecutive point pairs with gaps greater than $r_{k}$, which contradicts the definition of $r_{k}$ as the $k$-th largest gap between consecutive points. 
Let $i_{\max} := \argmax_{j : (p_j, q_j) \in \mathcal{P}} p_j$,
and let $q_{\max} := q_{i_{\max}}$, i.e., $q_{\max}$ is the larger point in the pair whose smaller point is the largest among all consecutive pairs.
Then define $C_{\mathcal{P}} := \{ p_i : (p_i, q_i) \in \mathcal{P} \}\cup \{q_{\max}\}$, so $|C_{\mathcal{P}}|\leq k$. 

We analyze the cost of this $C_{\mathcal{P}}$. 
Let $x_{\mathrm{end}}$ be the last point in $P'$, and let $X:= \{x\in P': q_{\max} < x < x_{\mathrm{end}} \}$ be the set of all points in $P'$ that lie between $q_{\max}$ and $x_{\mathrm{end}}$, then $|X|\leq n$. 
Write $X := \{x_1, x_2, \ldots \}$
such that $x_{i}< x_{i+1}$ for all $i\in [|X| - 1]$. 
Then, by triangle inequality, we have that $\dist(x_{\mathrm{end}}, q_{\max})\leq \dist(q_{\max}, x_{1}) + \sum_{i=1}^{|X|-1}\dist(x_{i}, x_{i+1}) + \dist(x_{|X|}, x_{\mathrm{end}})\leq O(n\cdot r_{k})$,
where the last inequality holds by the definition of $\mathcal{P}$.  
This implies that for any point in $P'$ after $p_{\max}$, the distance to $C_{\mathcal{P}}$ is at most $O(n\cdot r_{k})$. 
Similarly, we conclude that for any point in $P'$ before $p_{\max}$, the distance to its nearest point in $C_{\mathcal{P}}$ is also at most $O(n\cdot r_{k})$. 
This implies that $\cost(P', C_{\mathcal{P}})\leq O(n\cdot r_{k})$. 

The final value we return is $E := \Theta(n^3 E')$, to make sure the returned value is at least $\OPT$.
This entire algorithm is straightforward to implement in MPC, and the complexity is dominated by sorting.
\end{proof}

\subsection{Finding Approximate Assignments}
\label{sec:assignment}
We would assume the context in the above proof, and discuss how to find an approximate clustering for \kCenter in low and high dimensions, respectively (i.e., for \Cref{thm:low_dim_2approx,thm:low_dim_bicrit,thm:kcenter_high_dim}).

\begin{lemma}[Assignment for \kCenter in low dimension]
    \label{lemma:assignment_low_dim}
    There is an MPC algorithm, that given $\varepsilon\in (0, 1)$, $\tau > 0$,
    a center set $C$ and a dataset $P$ of $n$ points in $\mathbb{R}^{d}$ distributed across MPC machines with local memory $s\geq \Omega(\varepsilon^{-1}d)^{\Omega(d)}\poly(\log n)$, such that $\cost(P, C) = \Theta(\tau)$, computes an assignment for each point $x\in P$, mapping it to a center $c\in C$ such that $\dist(x, c)\leq (1 + O(\varepsilon))\cost(P, C)$, in $O(\log_{s}n)$ rounds using total memory $O(n\poly(d\log n))\cdot O(\varepsilon^{-1}d)^{O(d)}$. 
\end{lemma}
\begin{proof}
    Write $C := (c_{1}, \ldots, c_{k})$. 
Let $\phi : P \cup C \to \mathcal{G}_{\epsilon \tau / \sqrt{d}}$,
    that maps each point in $P$ or $C$ to its nearest neighbor in the $\epsilon \tau / \sqrt d$-grid $\mathcal{G}_{\epsilon \tau / \sqrt d}$.
    Let $P_{\tau} := \phi(P)$ and $C_{\tau} := \phi(C)$ be the sets rounded to the grid.
Before defining the clustering for the original $P$ based on $C$, 
    our first step finds the \emph{exact} nearest neighbor assignment from $P_{\tau}$ to $C_{\tau}$.
    The details of this part will be discussed later. 
    Now, assuming that the exact nearest neighbor assignment from $P_{\tau}$ to $C_{\tau}$ is obtained,
    we proceed to define how to assign points in $P$ to $C$ and prove that this assignment satisfies the approximation bound.
    
    For every $z \in \phi(C)$, pick $\rep(z)$  as an arbitrary fixed point in $\phi^{-1}(z) \cap C$, i.e., a representative point in $C$ among points whose nearest neighbor in $\mathcal{G}_{\epsilon \tau / \sqrt d}$ is the same grid.
    Now, for a point $x \in P$,
let $C_\tau(\phi(x))$ denote the nearest neighbor of $\phi(x)$ in $C_{\tau}$ (which is known by assumption),
    then we assign $x$ to $\rep( C_\tau(\phi(x)) )$.
    This step can be implemented within the claimed round and space as in \Cref{lemma:assignment_low_dim}.

    In this assignment, by triangle inequality,
    \begin{align*}
        \dist(x, C)
        &\leq \dist(x, \phi(x)) + \dist(\phi(x), C_\tau(\phi(x))) + \dist(C_\tau(\phi(x)), \rep(C_\tau(\phi(x)))) \\
        &\leq \varepsilon\tau + \dist(\phi(x), C_{\tau}) + \varepsilon\tau \\
        &\leq 2\epsilon \tau + \dist(x, C) + 2\epsilon \tau \\
        &\leq (1+O(\varepsilon))\cost(P, C).
    \end{align*}
    This finishes the description of the assignment from $P$ to $C$ and the analysis of its approximation bound.

    Now, we turn to define the exact nearest neighbor assignment from $P_{\tau}$ to $C_{\tau}$.  
    The algorithm simply computes, for every $x \in P_{\tau}$ the nearest center point $c\in B(x, O(\tau))\cap C_{\tau}$, i.e., only searching the nearest point in an $O(\tau)$ neighborhood of $x$. 
    This algorithm is well-defined,
    since for every $x \in P_\tau$,
    $\dist(x, C_\tau) \leq \dist(x, C) + \epsilon \tau \leq 2\epsilon \tau + \cost(P, C) \leq O(\tau)$. 
    This ensures that the found $c$ is the exact nearest neighbor of $x$.
Observe that for every $x \in P_\tau$, we have $|B(x, O(\tau)) \cap C_\tau| \leq (\epsilon^{-1}d)^{O(d)}$,
    by \Cref{lemma:packing}.
Hence, this step of finding the nearest neighbors of $P_\tau$ in $C_\tau$ 
    can be done in MPC in $O(\log_s n)$ rounds, $(\epsilon^{-1}d)^{O(d)} \poly(\log n)$ local space and $(\epsilon^{-1}d)^{O(d)} n \poly(\log n)$ total space.
    This finishes the proof.
\end{proof}

\begin{lemma}[Assignment for \kCenter in high dimension]
    \label{lemma:assignment_high_dim}

    There is an MPC algorithm, that given $\varepsilon\in (0, 1)$, $\tau > 0$,
    a center set $C$ and a dataset $P$ of $n$ points in $\mathbb{R}^{d}$ distributed across MPC machines
    with local memory $s\geq \poly(d\log n)$, such that $\cost(P, C) = \Theta(\tau)$, with probability at least $1-1/n$, computes an assignment for each point $x\in P$, mapping it to a center $c\in C$ such that $\dist(x, c)\leq O(\varepsilon^{-1})\cost(P, C)$, in $O(\log_{s}n)$ rounds using total memory $O(n^{1+\varepsilon}\poly(d\log n))$. 
\end{lemma}
\begin{proof}
We make use of a procedure described in \cite[Section 3.1]{CGJKV24}.  
    Specifically, in \cite[Section 3.1]{CGJKV24}, for a given parameter $r>0$, their MPC algorithm computes
    for each data point $x\in P$ an approximate center point $c_x \in C$
    such that $\dist(x, c_x) \leq O(\epsilon^{-1}) r$ if there exists $c \in C$ such that $\dist(x, c) \leq r$,
    in $O(\log_s n)$ rounds, $\poly(d\log n)$ local space and $n^{1 + \epsilon} \poly(d \log n)$ global space.
    
    Now, in our case, since $\cost(P, C) = \Theta(\tau)$, then there must exist a point $c \in C$ for each point $x\in P$  within a distance of $O(\tau)$. 
    Hence, applying the abovementioned algorithm with $r:= O(\tau)$, we can find for each $x\in P$ a center point $c_x\in C$ such that $\dist(x, c_x)\leq O(\varepsilon^{-1})\cost(P, C)$,
    and achieve the claimed round and space complexity as in \Cref{lemma:assignment_high_dim}.
\end{proof}

\bibliographystyle{alphaurl}
    \bibliography{ref.bib}

\appendix

    \section{Proof of Lemma~\ref{lemma:hash}: Geometric Hashing}
\label{sec:proof_hashing}

\lemmahash*

The construction of the hash $f: \mathbb{R}^{d}\to \mathbb{R}^{d}$ is the same as that  in \cite[Theorem 5.3]{arxiv.2204.02095}.
Their proof already implies the space bound, as well as the first two properties;
In fact, our second property is stronger than what they state, 
but the stronger version was indeed proved in their paper, albeit implicitly.
We would restate their algorithm/construction, briefly sketch the proof for the first two properties,  and focus on the third property.

\paragraph{Hash construction.}
We start with reviewing their algorithm/construction of $f$.
Partition $\mathbb{R}^{d}$ into hypercubes of side length $z$ of the form $\bigtimes_{i = 1}^d [a_i z, (a_i+1)z)$ where $a_1, \ldots, a_d \in \ZZ$. 
Notice that each of hypercubes is of diameter $\ell$. 
Let $\ell_{i}:= (d-i)\beta$ where $i\in\{0, \ldots, d-1\}$. 
For every $i\in\{0, \ldots, d\}$, let $\mathcal{F}_{i}$ be the set of $i$-dimensional faces of the abovementioned hypercubes,
and let $A_{i}$ be the set of points that belong to these faces, i.e., $A_i := \bigcup_{Q \in \mathcal{F}_i} Q$.
For each $i\in\{0, \ldots, d-1\}$, let $B_{i}:= N_{\ell_{i}}^{\infty}(A_{i})$ be $\ell_{i}$-neighborhood of $A_{i}$ (with respect to $\ell_\infty$ distance). 
Let $B_{-1}:= \emptyset$ and $B_{\leq i}:= \bigcup_{j\leq i}B_{j}$. 
The main procedure for the construction of $f$ goes as follows.
\begin{itemize}
\item For every $i\in\{0, \ldots, d-1\}$, for every $Q\in \mathcal{F}_{i}$, let $\widehat{Q} := N_{\ell_{i}}^{\infty}(Q)\setminus B_{\leq i-1}$, and for every $x\in\widehat{Q}$ assign $f(x)=q$ where $q\in\widehat{Q}$ is an arbitrary but fixed point. 
    Observe that $\widehat{Q}\subseteq B_{i}$ and thus $x\in\widehat{Q}$ will not be assigned again in later iterations. 
    \item For every $Q\in \mathcal{F}_{d}$, let $\widehat{Q} := Q\setminus B_{\leq d-1}$ be the remaining part of $Q$ whose $f(x)$ has not been assigned, and assign $f(x)=q$ for every $x\in\widehat{Q}$, where $q\in\widehat{Q}$ is arbitrary but fixed point. 
\end{itemize}  

\paragraph{First two properties.}
The first property is immediate from this construction since every bucket is a subset of a hypercube of side-length $z$ whose diameter is $\sqrt{d} z = \ell$.
To establish the second property, we need to define the $d + 1$ groups.
For $0\leq i\leq d-1$, define $\mathcal{W}_{i}:= \{N_{\ell_{i}}^{\infty}(Q)\setminus B_{\leq i-1}: Q\in \mathcal{F}_{i} \}$ and let $\mathcal{W}_{d}:=\{Q\setminus B_{\leq d-1}: Q\in \mathcal{F}_{d}\}$. 
Then by the construction of $f$, $\{f^{-1}(u): u\in f(\mathbb{R}^{d})\} = \bigcup_{i=0}^{d}\mathcal{W}_{i}$, hence the $\mathcal{W}_i$ is a proper partition of buckets.
These $\mathcal{W}_i$'s are implicitly analyzed in~\cite[Lemma 5.5 and Lemma 5.8]{arxiv.2204.02095}, restated and adapted (as \Cref{lemma:separation}) to our notation as follows.
This lemma readily implies our second property.
\begin{lemma}[{\cite[Lemma 5.5 and Lemma 5.8]{arxiv.2204.02095}}]
    \label{lemma:separation}
    For every $0 \leq i \leq d$ and every $S \neq S' \in \mathcal{W}_i$, $\dist(S, S') \geq \dist_\infty(S, S') > \beta$.
\end{lemma}

\paragraph{Third property.}
We proceed to prove the third property.
The third property, particularly $L(z, 2b)$ is well-defined, since by the choice of parameters it can be verified that $z > 4b$.
The proof starts with the following claim. It relates the complicated $\bigcup_{u} U_\tau^\infty(f^{-1}(u))$
to a much simpler object $N_b^\infty(A_{d - 1})$, which is merely the $b$-neighborhood of $(d - 1)$-dimensional faces.
\begin{claim}
    \label{claim:A_N} 
    $\bigcup_{u\in f(\mathbb{R}^{d})}U_{\tau}^{\infty}(f^{-1}(u))\subseteq N_{b}^{\infty}(A_{d-1})$. 
\end{claim}
\begin{proof}
Recall that $\{f^{-1}(u): u\in f(\mathbb{R}^{d})\} = \bigcup_{i=0}^{d}\mathcal{W}_{i}$.
Therefore, it suffices to prove that for every $i\in\{0, \ldots, d\}$ and for every $S\in \mathcal{W}_{i}$, $U_{\tau}^{\infty}(S)\subseteq N_{b}^{\infty}(A_{d-1})$. 
We prove this separately for the case of $i\leq d-1$ and $i=d$. 

\paragraph{Case I: $i \leq d - 1$}
We first analyze the case that $0 \leq i \leq d - 1$. 
Fix some $0 \leq i \leq d-1$
and some $S := N_{\ell_i}^\infty(Q) \setminus B_{\leq i - 1} \in\mathcal{W}_{i}$ (for some $Q \in \mathcal{F}_i$).
Observe that $S = N_{\ell_i}^{\infty}(Q)\setminus B_{\leq i-1} \subseteq N_{\ell_{i}}^{\infty}(Q)\subseteq N_{\ell_{i}}^{\infty}(A_{i})$. 
Then we have that $N_{\tau}^{\infty}(S)\subseteq N_{\tau}^{\infty}(N_{\ell_{i}}^{\infty}(A_{i})) = N_{\ell_{i}+\tau}^{\infty}(A_{i})$. 
Observe that $\ell_{i} = (d-i)\beta \leq d\beta$.
Hence, $U_\tau^\infty(S) = N_\tau^\infty(S) \setminus S \subseteq 
N_\tau^\infty(S) \subseteq N_{\ell_i + \tau}^\infty(A_i) \subseteq N_{d\beta + \tau}^\infty(A_{d - 1}) = N_b^\infty(A_{d - 1})$.
This finishes the case of $i \leq d - 1$.

\paragraph{Case II: $i = d$}
Next, we analyze the case that $i=d$.  
By \Cref{lemma:separation}, for every $S \neq S' \in \mathcal{W}_d$,
$\dist_\infty(S, S') > \beta \geq \tau$,
and this implies
\begin{equation}
    \label{eqn:empty_intersection}
    N_\tau^\infty(S) \cap S' = \emptyset,
\end{equation}
as $\ell_\infty$ distance between any two points is at most their $\ell_2$ distance.
Now fix some $S \in \mathcal{W}_d$.
By the construction, $A_d = \mathbb{R}^d$ and $B_{\leq d - 1} \cap \bigcup_{S \in \mathcal{W}_d} S = \emptyset$.
By \eqref{eqn:empty_intersection}, $N_\tau^\infty(S)$ has no intersection with  $\mathcal{W}_d$ other than on $S$.
Hence,
\begin{equation}
    \label{eqn:U_subset_B}
    U_\tau^\infty(S) = N_\tau^\infty(S) \setminus S \subseteq B_{\leq d - 1}.
\end{equation}
To further relate $B_{\leq d - 1}$ with $A_{d - 1}$, we have
\begin{align*}
    B_{\leq d - 1}
    = \bigcup_{j  = 0}^{d - 1} B_j
    = \bigcup_{j = 0}^{d  - 1}N_{\ell_j}^\infty(A_j)
    \subseteq \bigcup_{j = 0}^{d- 1} N_{d\beta}^\infty(A_j)
    \subseteq N_{d\beta}^\infty(A_{d - 1}),
\end{align*}
where the last step follows from $A_j \subseteq A_{j + 1}$ for every $j$.
Combining this with \eqref{eqn:U_subset_B}, we conclude that $U_\tau^\infty(S) \subseteq N_{d\beta}^\infty(A_{d - 1}) \subseteq N_{b}^\infty(A_{d - 1})$.
This finishes the proof of \Cref{claim:A_N}.
\end{proof}

By \Cref{claim:A_N}, it remains to prove $N_b^\infty(A_{d - 1}) \cap L(z, 2b)^d = \emptyset$.
Since $N_b^\infty(A_{d - 1}) = N_b^\infty(\bigcup_{Q \in \mathcal{F}_{d - 1}} Q)
= \bigcup_{Q \in \mathcal{F}_{d - 1}} N_b^\infty(Q) = \bigcup_{x \in Q, Q \in \mathcal{F}_{d - 1}}N_b^\infty(x)$,
it suffices to show $\forall x \in Q$ (for $ Q \in \mathcal{F}_{d - 1})$,
$N_b^\infty(x) \cap L(z, 2b)^d = \emptyset$.
Now fix some $Q \in \mathcal{F}_{d - 1}$ and $x \in Q$.
Since $Q$ is a $(d - 1)$-dimensional face, it has $d - 1$ dimensions spanning an entire interval of length $z$, and has exactly one dimension taking a value of the form $az$ ($a \in \ZZ$).
Formally, $x \in Q$ if and only if there exists
$0 \leq j \leq d$ and $\{a_i \in \ZZ\}_{i = 0}^d$ such that for $0 \leq i \leq d$,
\begin{equation}
    \label{eqn:point_d_1_face}
    x_i \in \begin{cases}
        [a_i z, (a_i + 1) z) & i \neq j \\
        \{ a_i z \} & i = j
    \end{cases}.
\end{equation}
Let $j$ be that as in \eqref{eqn:point_d_1_face}. Then for every $y \in L(z, 2b)^d$, 
\begin{align*}
    \|x - y\|_\infty
    \geq |x_j - y_j|
    \geq 2b,
\end{align*}
where the last inequality follows from the definition of $L(z, 2b)$ and recall that $L(z, 2b) = \bigcup_{a \in \ZZ} [az + 2b, (a + 1)z - 2b)$.
This further implies $\forall x' \in N_b^\infty(x)$ and $y \in L(z, 2b)^d$,
$\|x' - y\|_\infty \geq b$,
by the triangle inequality of $\ell_\infty$.
Therefore, we conclude that $N_b^\infty(x) \cap L(z, 2b)^d = \emptyset$,
and this finishes the proof of \Cref{lemma:hash}.
\qed

     \section{Analysis of One-round Luby's Algorithm}
\label{sec:one_round_luby}

Here we show both an upper and a matching lower bound for one-round Luby's algorithm for graphs.
We give a proof sketch, in \Cref{sec:luby_graph_ub}, for an upper bound that the set $S$ returned by the one-round Luby's algorithm is $O(\log n)$-ruling set with high probability.
We give a tight input graph $G$, such that the $S$ returned by the one-round Luby's algorithm is $\Omega(\log n)$-ruling set with high probability, in \Cref{sec:luby_graph_lb}.

Here, for an undirected graph $G = (V, E)$, 
an $\alpha$-ruling set is a subset $S \subseteq V$ such that a) for any $x \neq y \in S$, $\{x, y\} \notin E$, and b)
any vertex $x \in V$ has a path within $\alpha$ hops to $S$.
The one-round Luby's algorithm picks a uniform random value $h(x) \in [0, 1]$ for every $x \in V$,
and include $x$ into $S$ if and only if $x$ has the smallest $h$ value among $x$ and its adjacent vertices, i.e., $\{x\} \cup \{ y \in V : \{x, y \} \in E \}$.
Return $S$ as the output.

\subsection{Upper Bound: A Proof Sketch}
\label{sec:luby_graph_ub}
We discuss how the proof in \Cref{sec:ruling_set} can be modified to prove the $O(\log n)$-ruling set bound in graphs.
Even though we still use the language of Euclidean spaces in the following,
the proof sketch actually does not use Euclidean property and works directly for graphs.

As we are talking about one-round Luby's algorithm, we do not need to do the preprocessing in \Cref{alg:rs_hash},
and the in \Cref{alg:local_kcenter} we do not need to use the approximate $A_P^\beta(\cdot, \cdot)$.

The proof of \Cref{lemma:C_independent_set} still shows the returned point set is $\tau$-independent.
\Cref{lemma:rs_Pprime_to_P,fact:Lambda_ub} are not needed, as we work on $P$ directly (without preprocessing).
\Cref{lemma:approximation_ratio} is still the main lemma,
and the statement can be simply changed to for $t = O(\log n)$,
$\Pr[\forall p \in P, \dist(p, R) \leq O(t) \tau ] \geq 1- 1 / \poly(n)$.
The proof plan of \Cref{lemma:approximation_ratio} still goes through,
and it reduces to \Cref{lemma:bound_T},
where the statement is changed to for $t = O(\log n)$,
$\Pr[T \geq t] \leq 1 / \poly(n)$.

The main modification is in the proof of \Cref{lemma:bound_T}.
\Cref{lemma:ext_cond} is still useful even in the graph setting,
and does not need to be changed.
However, we use a new definition for $f$, and instead of mapping to $\mathbb{Z}$,
we do $f : \mathcal{S} \to \{0, 1\}$.
Pick some parameter $\gamma = \Theta(1)$.
For $S' \in \mathcal{S}$, define 
\begin{align*}
    f(S') := \begin{cases}
        1 & \frac{|\Anew(S')|}{|\Anew(S') \cup \Acond(S')|} > \gamma \\
        0 & \text{otherwise}
    \end{cases}
\end{align*}
Intuitively, $f(S')$ is $1$ if the upper bound $\frac{|\Anew(S')|}{|\Anew(S') \cup \Acond(S')|}$ of the extension probability (\Cref{lemma:ext_cond}) is too large.
Ideally, if for a sequence all extensions have $f$ value $0$, then its length is greater than $t$ with probability $1 / \poly(n)$.
The definition of \emph{configuration} remains the same, i.e., for $S' \in \mathcal{S}$,
i.e., $\conf(S'):= (f(S'_{[0,1]}), \ldots, f(S'_{[0, m]}))$.

Next, we need a new lemma that shows for every $\pi \in \{0, 1\}^t$ and $S' \in \mathcal{S}_\pi$,
$\|\conf(S') \|_1 \leq t / 2$ (which requires the constant in the big-O of $t$ to be picked carefully).
This lemma holds because once $f$ value is $1$, the $|\Anew|$ (which is the new part introduced by the extension) is at least constant fraction of the neighborhood of all previous elements in the sequence.
Therefore, this cannot happen more than $t / 2 = \Omega(\log n)$ times since otherwise $|\Anew(S'_{[0, t / 2]})| > n$.

Finally, in the modified version of \Cref{lemma:key},
we use the abovementioned $\|\conf(S')\|_1 \leq t / 2$ bound,
and conclude that there are still at least $t / 2 = \Omega(\log n)$ number of extensions
that happen with probability at most  $\frac{|\Anew(S')|}{|\Anew(S') \cup \Acond(S')|} \leq \gamma$,
and hence $\sum_{S' \in \mathcal{S}_\pi}\Pr[S' \sqsubseteq S] \leq \gamma^{t / 2}$.

Finally, since the number of configurations is at most $2^t$,
taking a union bound, $2^t \cdot \gamma^{t / 2} \leq 1 / \poly(n)$ by setting small enough $\gamma$.

\subsection{Lower Bound}
\label{sec:luby_graph_lb}
Let $n$ be some parameter.
Consider the following graph $G = (V, E)$.
Let $m := \frac{\ln n}{\gamma}$ where $\gamma$ is a sufficiently large constant.
The vertex set $V$ consists of $m$ parts $V_1, \ldots, V_m$, i.e., $V := \bigcup_i V_i$,
such that $|V_i| = 2^i$.
We add for every $i \in [m]$, a clique in each $V_i$, i.e., $V_i \times V_i$,
and add for every $i \in [m - 1]$
a complete bipartite graph for between $V_i$ and $V_{i + 1}$, i.e., $V_i \times V_{i + 1}$.

Now, consider the one-round Luby's algorithm on $G$,
and we analyze the random values $h$.
For $i \in [m - 1]$, let $\mathcal{E}_i$ be the event that there exists $u \in V_{i + 1}$ such that
$h(u) < h(v)$ for all $v \in V_{i}$.

\begin{claim*}
    $\Pr[\bigwedge_{i \in [m-1]} \mathcal{E}_i] \geq 1 / \sqrt{n}$.
\end{claim*}
\begin{proof}
    \begin{align*}
        \Pr[\bigwedge_{i \in [m-1]} \mathcal{E}_i] 
        &= \prod_{i \in [m-1]} \Pr[\mathcal{E}_i \mid \bigwedge_{j \in [i - 1]} \mathcal{E}_j] \\
        &= \prod_{i \in [m-1]} \Pr[\mathcal{E}_i] \\
        &= \prod_{i \in [m-1]} \left(1 - \left(1 - \frac{1} {2^i + 1}\right)^{2^{i + 1}}\right) \\
        &\geq (1 - 1 / e^2)^m \\
        &\geq 1 / \sqrt n.
    \end{align*}
    To see the second equality,
consider $\Pr[\mathcal{E}_{i} \mid \bigwedge_{j\in [i-1]} \mathcal{E}_{j}]$ for some $i\in [m - 1]$.
    Then the values $h$ for the vertices in $V_{i + 1}$ and $\bigcup_{j\in [i]} V_{j}$ are independent, and that the event $\bigwedge_{j\in [i-1]} \mathcal{E}_{j}$ only depends on the randomness of $h$ for the vertices in $\bigcup_{j\in [i]} V_{j}$. 
    The third equality follows from the fact that the $h$ values for vertices in $V_{i+1}$ are independent to each other,
    and that for a fixed vertex $u \in V_{i + 1}$ the probability that $h(u) < h(v)$ for every $v \in V_i$
    is $1 / (2^i + 1)$.
\end{proof}
Observe that the event $\bigwedge_{i \in [m-1]} \mathcal{E}_{i}$ implies that one-round Luby's returns $\Omega(\log n)$-ruling set.

Hence, we define a new graph $G'$ by duplicating $G$ for $n^{0.6}$ times (where there are no edges between the copies).
This graph $G'$ has $\poly(n)$ vertices.
Moreover, we know that with probability at least $1 - (1 - 1 / \sqrt{n})^{n^{0.6}} > 0.5$,
one-round Luby's algorithm on $G'$ returns an $\Omega(\log n)$-ruling set.
This finishes the proof.

\end{document}